\documentclass{amsproc}
\usepackage{cite}
\usepackage{amsmath,amsfonts,amssymb,amsthm,amscd,mathrsfs,mathdots,bbm,color,bm}
\numberwithin{equation}{section}
\usepackage{enumerate}
\usepackage{comment}

\usepackage{graphicx}
\usepackage[linecolor=white,backgroundcolor=orange,bordercolor=red]{todonotes}

\usepackage{newtxmath}
\usepackage{newtxtext}
\usepackage{subcaption}

\usepackage{ stmaryrd }

\DeclareFontFamily{U}{mathx}{}
\DeclareFontShape{U}{mathx}{m}{n}{<-> mathx10}{}
\DeclareSymbolFont{mathx}{U}{mathx}{m}{n}
\DeclareMathAccent{\widecheck}{0}{mathx}{"71}

\def\R{{\mathbb R}}
\def\C{{\mathbb C}}
\def\N{{\mathbb N}}
\def\Z{{\mathbb Z}}

\def\AA{\mathcal{A}}
\def\BB{\mathcal{B}}
\def\FF{\mathcal{F}}
\def\SS{\mathcal{S}}

\def\sine{\operatorname{sine}}
\def\Ai{\operatorname{Ai}}

\def\sigmaPc{{\sigma}_{P_{<N}}^{\hbar}}
\def\sigmaHc{{\sigma}_{H_N}^{\hbar}}
\def\sigmaKc{{\sigma}_{P_{<N}}^{\hbar}}
\def\sigmaQc{{\sigma}_{H_N}^{\hbar}}

\newtheorem{thm}{Theorem}

\newtheorem{prop}{Proposition}
\newtheorem{lem}{Lemma}
\newtheorem{definition}{Definition}

\theoremstyle{definition}
\newtheorem*{notation}{Notation}
\theoremstyle{remark}
\newtheorem{rem}{Remark}

\newcommand{\be}{\begin{equation}}
\newcommand{\ee}{\end{equation}}

\newcommand{\ben}{\begin{equation*}}
\newcommand{\een}{\end{equation*}}

\def\_#1{\def\next{#1}%
 \ifx\next\risingsign\expandafter\rising\else^{\underline{#1}}\fi}
\def\risingsign{^}
\def\rising#1{^{\overline{#1}}}

\title[The semiclassical limit of a quantum Zeno dynamics]{The semiclassical limit of a quantum Zeno dynamics}

\author[F.~D.~Cunden]{Fabio Deelan Cunden}
\address{Dipartimento di Matematica, Univerisit\`a degli Studi di Bari, I-70125 Bari, Italy, and
INFN, Sezione di Bari, I-70126 Bari, Italy
}
\email{fabio.cunden@uniba.it}
\author[P.~Facchi]{Paolo Facchi}
\address{Dipartimento di Fisica, Univerisit\`a degli Studi di Bari, I-70126 Bari, Italy, and
INFN, Sezione di Bari, I-70126 Bari, Italy
} 
\email{paolo.facchi@ba.infn.it}
\author[M.~Ligab\`o]{Marilena Ligab\`o}
\address{Dipartimento di Matematica, Univerisit\`a degli Studi di Bari, I-70125 Bari, Italy}
\email{marilena.ligabo@uniba.it}

\begin{document}
 
\maketitle
  
\begin{abstract}  Motivated by a quantum Zeno dynamics in a cavity quantum electrodynamics setting, we study the asymptotics of a family of symbols corresponding to a truncated momentum operator, in the semiclassical limit of vanishing Planck constant $\hbar\to0$ and large quantum number $N\to\infty$, with $\hbar N$ kept fixed. In a suitable topology, the limit is the discontinuous symbol $p\chi_D(x,p)$ where $\chi_D$ is the characteristic function of the classically permitted region $D$ in  phase space. A refined analysis shows that the symbol is asymptotically close to the function $p\chi_D^{(N)}(x,p)$, where $\chi_D^{(N)}$ is a smooth version of $\chi_D$ related to the integrated Airy function. We also discuss the limit from a dynamical point of view.

\end{abstract}
  
  \section{Introduction}
   \label{sec:Introduction}
   
In the \emph{quantum Zeno effect}, frequent projective measurements can slow down the evolution of a quantum system and eventually hinder any transition to states different from the initial one. The situation is much richer when the measurement does not confine the system in a single state, but rather in a multidimensional subspace of its Hilbert space. This gives rise to a \emph{quantum Zeno dynamics} (QZD): the system evolves in the projected subspace under the action of its projected Hamiltonian. 
This phenomenon, first considered by Beskow and Nilsson~\cite{Beskow} in their study of the decay of unstable systems, was dubbed quantum Zeno effect (QZE) by Misra and Sudarshan~\cite{Misra} who suggested a parallelism with the paradox of the `flying arrow at rest' by the philosopher Zeno of Elea.
Since then, QZE has received constant attention by physicists and mathematicians, who explored different aspects of the phenomenon.

From the mathematical point of view, QZD is related to the limit of a product formula obtained by intertwining the dynamical time evolution group with the orthogonal projection associated with the measurements performed on the system. It can be viewed as a generalization of Trotter-Kato product formulas~\cite{chernoff, kato, Trotter1, Trotter2} to more singular objects in which one semigroup is replaced by a projection.
The structure of the QZD product formula has been thoroughly investigated and has been well characterized under quite general assumptions~\cite{exner1, exner2, exner3, exner4, friedman1, friedman2, gustafson, Matolcsi, Schmidt1, Schmidt2, Facchi10}.

QZE has been observed experimentally in a variety of systems, on experiments involving photons, nuclear spins, ions, optical pumping, photons in a cavity, ultracold atoms, and Bose-Einstein condensates, see~\cite{zenoreview} and references therein. In all the abovementioned implementations, the quantum system is forced to remain in its initial state through a measurement associated with a one-dimensional projection. 
The present study is inspired by a  proposal  by Raimond \emph{et al.}~\cite{Raimond10,Raimond12} for generating a multidimensional QZD in a cavity quantum electrodynamics experiment. We briefly describe the proposal, skipping most of the non-mathematical details.

The mode of the quantized electromagnetic field in a cavity can be conveniently described in the Fock space representation. The Hamiltonian of the quantized field is that of a harmonic oscillator (with angular frequency $\omega=1$)
\begin{gather}
H_{\operatorname{h.o.}}=\frac{1}{2}\left(-\hbar^2\frac{d^2}{d x^2}+\hat{x}^2\right),
\label{eq:Hho}
\end{gather}
where $\hat{x}$ is the position operator and $\hat{p}=-i\hbar \frac{d}{dx}$ is the momentum operator on $L^2(\R)$. The operators $\hat{x}$, $\hat{p}$, and $H_{\operatorname{h.o.}}$ are essentially self-adjoint on the common core $\SS(\R)$, the Schwartz space of rapidly decreasing functions. The eigenfunction $\psi_n$ of $H_{\operatorname{h.o.}}$ represents  a cavity state with $n$ photons ($n=0,1,2,\dots$) and energy $\lambda_n= \hbar (n+1/2)$.

The cavity field undergoes a stroboscopic evolution alternating a short
continuous time evolution $e^{-i \frac{\tau}{\hbar} \hat{p}}$ given by a \emph{displacement operator}, that without loss of generality is taken to be generated by  $\hat{p}$, and an instantaneous interaction 
\be \label{eqn:PN}
P_{<N}=\sum_{k=0}^{N-1}|\psi_k\rangle\langle \psi_k|=\chi_{(-\infty,\hbar N)}(H_{\operatorname{h.o.}}),
\ee 
with atoms injected into the cavity to ascertain whether in the cavity there are less than $N$ photons ($N \geq 1$ is a chosen maximal photon number). 

The quantum Zeno dynamics consists in performing a series of $P_{<N}$-measurements in a fixed time interval $[0,t]$ at times $t_j= j \tau$, $j=0, \dots,n$, with period $\tau=t/n$. 
The intertwining of the continuous time evolutions and the projective measurements corresponds to the evolution operator
$$
V_n(t)= \left(P_{<N}e^{- \frac{i t}{n\hbar}\hat{p}}P_{<N}\right)^n.
$$
Observe that since $\operatorname{Ran} P_{<N}\subset D(\hat{p})=H^{1}(\R)$, we have~\cite{Facchi10} 
$$
\lim_{n\to \infty}V_n(t)=P_{<N}e^{-itH_N/\hbar},
$$
in the strong operator topology, uniformly for $t$ in compact subsets of $\R$, where the \emph{Zeno Hamiltonian}  $H_N$ is a rank-$N$ truncation of $\hat{p}$:
\begin{align}
H_N&=P_{<N}\hat{p}P_{<N} \nonumber\\
&=\chi_{(-\infty,\hbar N)}(H_{\operatorname{h.o.}})\,\hat{p}\,\chi_{(-\infty,\hbar N)}(H_{\operatorname{h.o.}}) \label{eqn:HN}.
\end{align}

Hence the QZD establishes a sort of `hard wall' in the Hilbert space: the state of the system evolves unitarily within the $N$-dimensional \emph{Zeno subspace} spanned by states with at most $(N-1)$ photons, $\psi_0,\dots,\psi_{N-1}$. This hard wall prevents the state to escape from the Zeno subspace and induces remarkable features in the quantum evolution~\cite{Raimond10,Raimond12}.  

The question addressed in this paper is: 
\emph{What is the semiclassical limit of the Zeno Hamiltonian $H_N$ and of its corresponding quantum dynamics?}

\subsection{Semiclassical limit of the Zeno Hamiltonian}\label{sect:zenoham}

Semiclassical theory concerns the asymptotic analysis for vanishing Planck constant ($\hbar\to0$) of operators and vectors,  with the ultimate goal of understanding the quantum-to-classical transition. It is therefore convenient to use a phase space description of quantum mechanics where  operators are represented by functions on the classical phase space (called Weyl symbols), states are described by quasi-probability distributions (called Wigner functions), and the noncommutative product of operators is mapped in a twisted convolution product of symbols (called Moyal product), see e.g.~\cite{folland, robert}. 

If we describe the QZD in the phase space, in the semiclassical limit $N\to\infty$, $\hbar\to0$ with the product $\hbar N=\mu$ kept fixed, we expect the motion to be confined in the classically allowed region. The level sets of the classical harmonic oscillator 
\be
\mathfrak{h}_{\operatorname{h.o.}}(x,p)=\frac{1}{2}\left(p^2+x^2\right)
\ee 
are circles centered at the origin of the phase space $\R_x\times\R_p$ . In qualitative terms, the hard wall can be viewed in the phase space as a circle with a radius $\propto\sqrt{\hbar N}$. In the limit, the corresponding classically allowed region is the disk 
$$
D:=\{(x,p)\in\R^2\colon \mathfrak{h}_{\operatorname{h.o.}}(x,p) < \mu\}=\{p^2+x^2< 2\mu\},$$ 
whose boundary
$$\partial D:=\{(x,p)\in\R^2\colon \mathfrak{h}_{\operatorname{h.o.}}(x,p)= \mu\}=\{p^2+x^2= 2\mu\}$$
is the circle  of radius $\sqrt{2\mu}$. This is what Raimond \emph{et al.}~\cite{Raimond10,Raimond12} called the `exclusion circle': it separates $D$  from the classically forbidden region where $\mathfrak{h}_{\operatorname{h.o.}}>\mu$.

\begin{figure}[t]
	\centering
	\includegraphics[width=.65\textwidth]{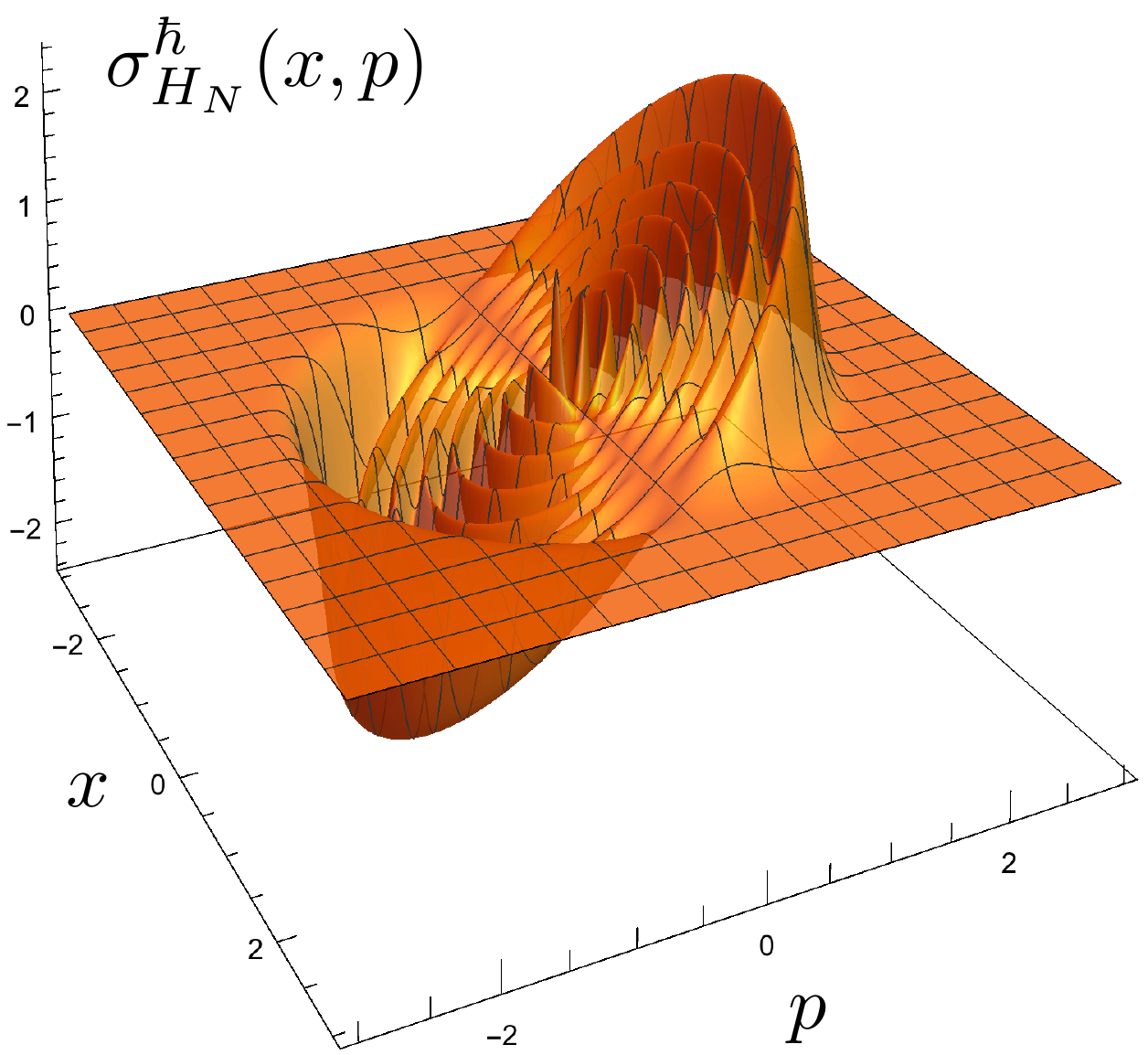}
	\caption{Plot of the symbol $\sigmaHc$ in the phase plane $(x,p)$. Here $N=17$ and $\mu=2$. Already for such a small value of $N$, the graph of the symbol resembles a (rippled) tilted coin in the disk $D$ and zero outside.}
	\label{fig:symbolQ}
\end{figure}

Let $\chi_D(x,p)=\chi_{(-\infty,\mu)}\big(\mathfrak{h}_{\operatorname{h.o.}}(x,p)\big)$ be the characteristic function of the disk $D$.
The first main result of the paper is the identification of the limit of the Weyl symbols $\sigmaPc(x,p)$ and $\sigmaHc(x,p)$ of the projection operator $P_{<N}$ and the Zeno Hamiltonian $H_N$, respectively. (The  definition of the Weyl symbol of an operator is given in Definition~\ref{def:Weyl}.) 

 \begin{thm}[Weak convergence of the symbols]\label{thm:convergenceAprime_intro} Set $\mu>0$. Then,
 \begin{align}
\label{eq:asymp_sigma_P_<N}
\lim_{\substack{N\to\infty, \hbar\to0\\\hbar N=\mu}}\int\limits_{\R_x\times\R_p}\left[\sigmaPc(x,p)-\chi_{D}(x,p)\right]\varphi(x,p)dxdp=0,\\
\label{eq:asymp_sigma_H_N}
\lim_{\substack{N\to\infty, \hbar\to0\\\hbar N=\mu}}\int\limits_{\R_x\times\R_p}\left[\sigmaHc(x,p)-p\chi_{D}(x,p)\right]\varphi(x,p)dxdp=0,
\end{align}
for all $\varphi\in\AA$. 
\end{thm}
\begin{rem}
Here  $\AA$ is the space of test functions introduced by Lions and Paul~\cite{Lions93} as the completion of the smooth functions of compact support in the phase space $C^{\infty}_c(\R_x\times\R_p)$ under the norm
\be
\|\varphi\|_{\AA}:=\int_{\R} dy\sup_x\left|\FF_2 \varphi(x,y)\right|.
\ee
In this paper $\FF_2 \varphi$ denotes the partial Fourier transform of $\varphi$ in the second variable,
\be
\FF_2 \varphi(x,y):=\int_{\R} \varphi(x,p)e^{-ipy}dp.
\ee
\end{rem}
\begin{table}
\begin{tabular}{l|ll}
Quantum & Classical \\
 $N\in\N$& $\hbar\to0$, $N\to\infty$\\
  $\hbar>0$&with $\hbar N=\mu$ &  \\
 \hline\hline\\
$P_{<N}=\chi_{(-\infty,\hbar N)}(H_{\operatorname{h.o.}})$ & $\chi_{(-\infty,\mu)}(\mathfrak{h}_{\operatorname{h.o.}}(x,p))=\chi_{(-\infty,\sqrt{2\mu})}(\sqrt{x^2+p^2})$ \\   \\
$H_N=\chi_{(-\infty,\hbar N)}(H_{\operatorname{h.o.}})\,\hat{p}\,\chi_{(-\infty,\hbar N)}(H_{\operatorname{h.o.}})$ & $p\chi_{(-\infty,\mu)}(\mathfrak{h}_{\operatorname{h.o.}}(x,p))=p\chi_{(-\infty,\sqrt{2\mu})}(\sqrt{x^2+p^2})$   \\   \\
\end{tabular}
\caption{Summary of the operators and their semiclassical limits.}
\label{eq:tab_summary}
\end{table}
\begin{rem}
Theorem~\ref{thm:convergenceAprime_intro} makes precise the heuristic expectation that the symbol $\sigmaPc(x,p)$ of the projection operator converges to the characteristic function $\chi_{D}(x,p)$ of the classically allowed region, and  the symbol of the Zeno Hamiltonian $\sigmaHc(x,p)$ converges to $p\chi_{D}(x,p)$. The content of Theorem~\ref{thm:convergenceAprime_intro} is schematically summarised in Table~\ref{eq:tab_summary}.   
\end{rem}
\begin{rem}
A  plot of the Weyl symbol exhibits pronounced oscillations, also known as \emph{quantum ripples}~\cite{Bettelheim12}, in  $ D=\{\mathfrak{h}_{\operatorname{h.o.}}(x,p)<\mu\}$. If the oscillations are smoothed out, then the graph of $\sigmaHc$ is asymptotically close to a `tilted coin'. See Fig. ~\ref{fig:symbolQ} and Fig.~\ref{fig:symbolQ2}. 
\end{rem}

\begin{figure}[t]
	\centering
	\includegraphics[width=.65\textwidth]{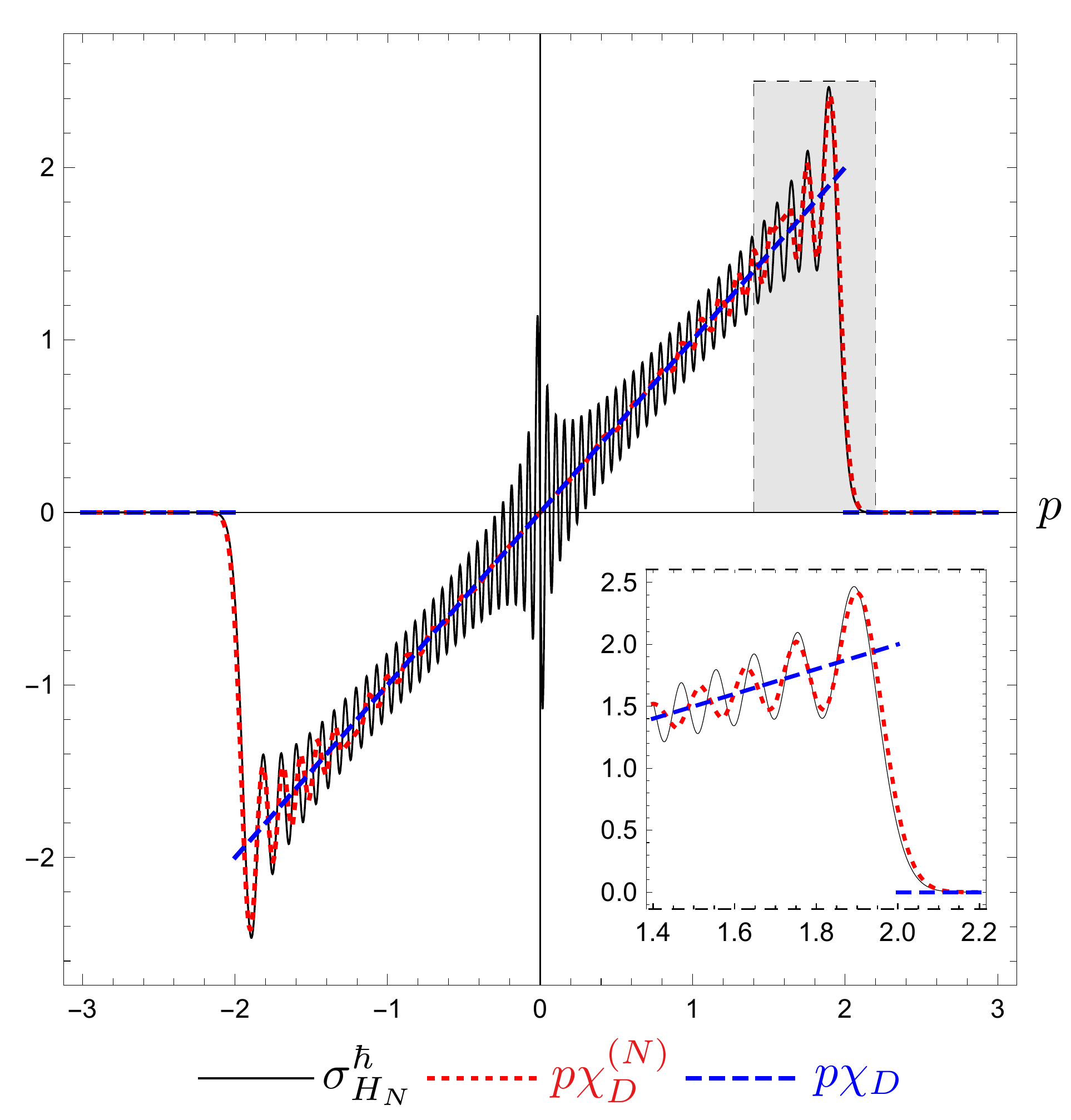}
	\caption{Comparison between the symbol of the Zeno Hamiltonian, $\sigmaHc$, and its semiclassical versions,  $p\chi_D^{(N)}$  and $p\chi_D$, as functions of $p$ with $x=0$ fixed. The inset (a zoom of the shaded area) show how $\chi_D^{(N)}$ better approximates the symbol $\sigmaHc$ near the edge. Here $N=57$, $\mu=2$.}
	\label{fig:symbolQ2}
\end{figure}

We see that at the boundary $\partial D$  the symbols $\sigmaPc$ and $\sigmaHc$ develop a jump, for large $N$. The second main result of the paper concerns a finer asymptotics of $\sigmaPc$ and $\sigmaHc$ near $\partial D$.
By zooming in at the edge $\partial D$, one sees that the symbols have nontrivial scaling limits related to the integrated  Airy function 
$$
\Ai_1(\xi):=\int_{\xi}^{+\infty} \Ai (u)\, du, \qquad \xi \in \R,
$$ 
see~\eqref{eq:Ai1}.  More precisely, set 
\be
\chi_{D}^{(N)}(x,p):=\Ai_1\left(\frac{(2N)^{\frac{2}{3}}}{\mu}\left(\mathfrak{h}_{\operatorname{h.o.}}(x,p)-\mu\right)\right), \qquad x,p \in \R. 
\ee 
It follows from~\eqref{eq:asymp_Ai1}  that $\chi_{D}^{(N)}$ is a sequence of rotational symmetric smooth functions on the phase space that approximate the characteristic function, 
$$
\lim_{N\to\infty}\chi_{D}^{(N)}(x,p)=\chi_D(x,p)
$$
for all $(x,p)\notin\partial D$. (On the boundary $\partial D$, $\chi_{D}^{(N)}=1/3$, for all $N$.)

We can now state our second main result.
\begin{thm}[Weak asymptotics at the boundary]\label{thm:convergenceAiry_intro}  Fix $\mu>0$. For all $g\in C^{\infty}_c(\R)$, 
\be
\label{eq:edge_asymp_K}
\lim_{\substack{N\to\infty, \hbar\to0\\\hbar N=\mu}}\int\limits_{\R_x\times\R_p} \left[\sigmaPc(x,p)-\chi_{D}^{(N)}(x,p)\right]\frac{1}{\hbar^{\frac{2}{3}}}g\left(\frac{x^2+p^2-2\mu}{\hbar^{\frac{2}{3}}}\right)dxdp=0.
\ee
and
\be
\label{eq:edge_asymp_Q}
\lim_{\substack{N\to\infty, \hbar\to0\\\hbar N=\mu}}\int\limits_{\R_x\times\R_p}\left[\sigmaHc(x,p)-p\chi_{D}^{(N)}(x,p)\right]\frac{1}{\hbar^{\frac{2}{3}}}g\left(\frac{x^2+p^2-2\mu}{\hbar^{\frac{2}{3}}}\right)dxdp=0.
\ee
\end{thm}

\begin{rem}
In order to zoom at $\partial D$, we need to integrate the symbols $\sigmaPc$ and $\sigmaHc$ against (sequences of) compactly supported test functions that concentrate around $\partial D$. Since $\partial D$ is invariant under rotations, without loss of generality we consider test functions that are also rotational symmetric. The idea is to consider, for $g\in C^{\infty}_c(\R)$, the rescaling $\epsilon^{-2}g(\epsilon^{-2}(x^2+p^2-2\mu))$ that is nonzero in a region of order $\mathrm{O}(\epsilon)$ within the boundary $\partial D$. The blow-up scale that gives rise to a nontrivial limit is $\epsilon=\hbar^{\frac{1}{3}}$. The reason for this choice will emerge in the following (see Section~\ref{sect:scaling}).
Note that the space of test functions $\AA$ in Theorem~\ref{thm:convergenceAprime_intro} does not depend on the details of the model. On the contrary, in Theorem~\ref{thm:convergenceAiry_intro} we integrate the symbols $\sigmaPc$ and $\sigmaHc$ against test functions that concentrate around $\partial D$ in a suitable way. 
\end{rem}

\begin{rem}\label{rempoint}
The limits in Theorems~\ref{thm:convergenceAprime_intro} and~\ref{thm:convergenceAiry_intro} do \emph{not} hold pointwise, in general. For instance, 
it is easy to show (by using the parity of the harmonic oscillator eigenfunctions) that
\begin{align}
\sigmaPc(0,0)
=1+(-1)^{N+1}.
\label{eq:sigma_in_0}
\end{align}
The reader is invited to have a glance at Fig.~\ref{fig:symbolQ2}. Inside the disk $D$, the symbols oscillate with frequency of order $\operatorname{O}(N)$, while in the classically forbidden region $D^c= (\R_x \times \R_p) \setminus D$ the symbols are exponentially suppressed. The monotonic behaviour outside the disk suggests that for $(x,p)\in D^c$ the convergence to the limits may hold in a stronger sense. In fact, a slight adaptation of the proof of Theorem~\ref{thm:convergenceAiry_intro} shows that outside the disk, $\sigmaPc$ and $\sigmaHc$  converge pointwise to the limit symbols. 
\end{rem}

\begin{thm}[Pointwise asymptotics in the classically forbidden region]
 \label{thm:edge_out_pointwise} Fix $\mu>0$.  For all $(x,p)\in D^c$,
\be
\label{eq:edge_pointwise1}
\lim_{\substack{N\to\infty, \hbar\to0\\\hbar N=\mu}} \left[\sigmaKc(x,p)-\chi_{D}^{(N)}(x,p)\right]=0,
\ee
and
\be
\label{eq:edge_pointwise2}
\lim_{\substack{N\to\infty, \hbar\to0\\\hbar N=\mu}}\left[\sigmaQc(x,p)-p\chi_{D}^{(N)}(x,p)\right]=0.
\ee
\end{thm}

\subsection{Semiclassical limit of the quantum Zeno dynamics}
\label{sec:Hamiltonian}

The quantum dynamics in  phase space is ruled by two elements: the Weyl symbol of the Zeno Hamiltonian  $\sigmaQc$ and the Moyal bracket (that does depend on $\hbar$)~\cite{folland, robert}. Hence, the semiclassical limit of the dynamics should encompass a simultaneous $\hbar\to0$ limit of the symbol (the generator of the dynamics) \emph{and} the Moyal structure.

By Theorem~\ref{thm:convergenceAprime_intro}, the symbol $\sigmaQc$  of the Zeno Hamiltonian converges as $\hbar\to0$, $N\to\infty$, with $\hbar N = \mu >0$, to
$p\chi_D(x,p)$. Moreover,  the Moyal bracket has an asymptotic expansion in powers of $\hbar$ whose leading term (zero-th order) is the classical Poisson bracket. Hence, it is reasonable to expect that the limiting dynamics is well described by the Hamiltonian evolution (i.e. Poisson) in  phase space where the Hamiltonian is the limit symbol $p\chi_D(x,p)$.

However, in this na\"ive approach we immediately face an obstruction: the symbol $p\chi_D(x,p)$ is \emph{not} smooth, and hence it is not possible to write Hamilton's equations of motion!
If we insist in writing, formally, Hamilton's equations, we get
\be
\begin{cases}
\displaystyle\dot{x}=\dfrac{\partial }{\partial p}\left(p\chi_D(x,p)\right)=\chi_{[0,\sqrt{2\mu})}(r)+p\delta_{\sqrt{2\mu}}(r)\dfrac{p}{r},\\
\displaystyle\dot{p}=-\dfrac{\partial}{\partial x}\left(p\chi_D(x,p)\right)=-p\delta_{\sqrt{2\mu}}(r)\dfrac{x}{r},
\end{cases}\qquad (\diamondsuit)
\ee 
where $r=\sqrt{x^2+p^2}$. The Dirac delta $\delta_{\sqrt{2\mu}}(r)$ arises as distributional derivative of the step function. We stress again that the above expressions are formal: the Hamiltonian is discontinuous at  $\partial D$, and its  vector field in  $(\diamondsuit)$ is singular.

We can now look at the corresponding phase portrait. First, the Hamiltonian vector field is zero outside the closure of the disk $D$. Thus, all points there are equilibrium points. 
If the particle is in $D$, then  the equation of motions are $\dot{x}=1,\dot{p}=0$, and the particle moves with constant momentum
$$
x(t)=x_0+t,\qquad p(t)=p_0.
$$
It is thus proceeding at a constant velocity along the $x$-axis. When it hits the boundary $\partial D$, the evolution is given by the singular contributions, proportional to the delta functions: 
$\dot{x}=p\delta_{\sqrt{2\mu}}(r)p/r,\dot{p}=-p\delta_{\sqrt{2\mu}}(r)x/r$. Heuristically, these equations would correspond to a field tangential to the boundary of $D$ that yields a motion along the circle $\partial D$ at `infinite' speed. The particle reappears on the other side of the boundary (with the same momentum $p = p_0$) and resumes its motion along the $x$-axis at a constant velocity. The collision at the edge $\partial D$ thus realizes, in this semiclassical picture, a reflection around the $p$-axis of the phase space, transforming $\left(\sqrt{2\mu-p_0^2},p_0\right)$ into $\left(-\sqrt{2\mu-p_0^2},p_0\right)$.

An interesting interpretation of the semiclassical limit of the Zeno dynamics is as follows. In the limit dynamics, the points $(x,p)$, $(-x,p)$ on the cirle $\partial D\subset\R_x\times\R_p$ are identified. Hence, one can think of the $N\to\infty,\hbar\to0$ limit, with $\hbar N = \mu$, as yielding a \emph{change of topology}: the dynamics on the disk becomes a motion on the sphere!
We emphasise again that all this is formal, although very close to what was observed in~\cite{Raimond10,Raimond12}, and called `phase inversion mechanism'. The function  $p\chi_D(x,p)$ is \emph{not} smooth and therefore it is not the generator of a classical Hamiltonian dynamics.

\begin{figure}[t]
	\centering
		\begin{subfigure}[t]{0.31\columnwidth}
		\centering
		\includegraphics[width=\textwidth]{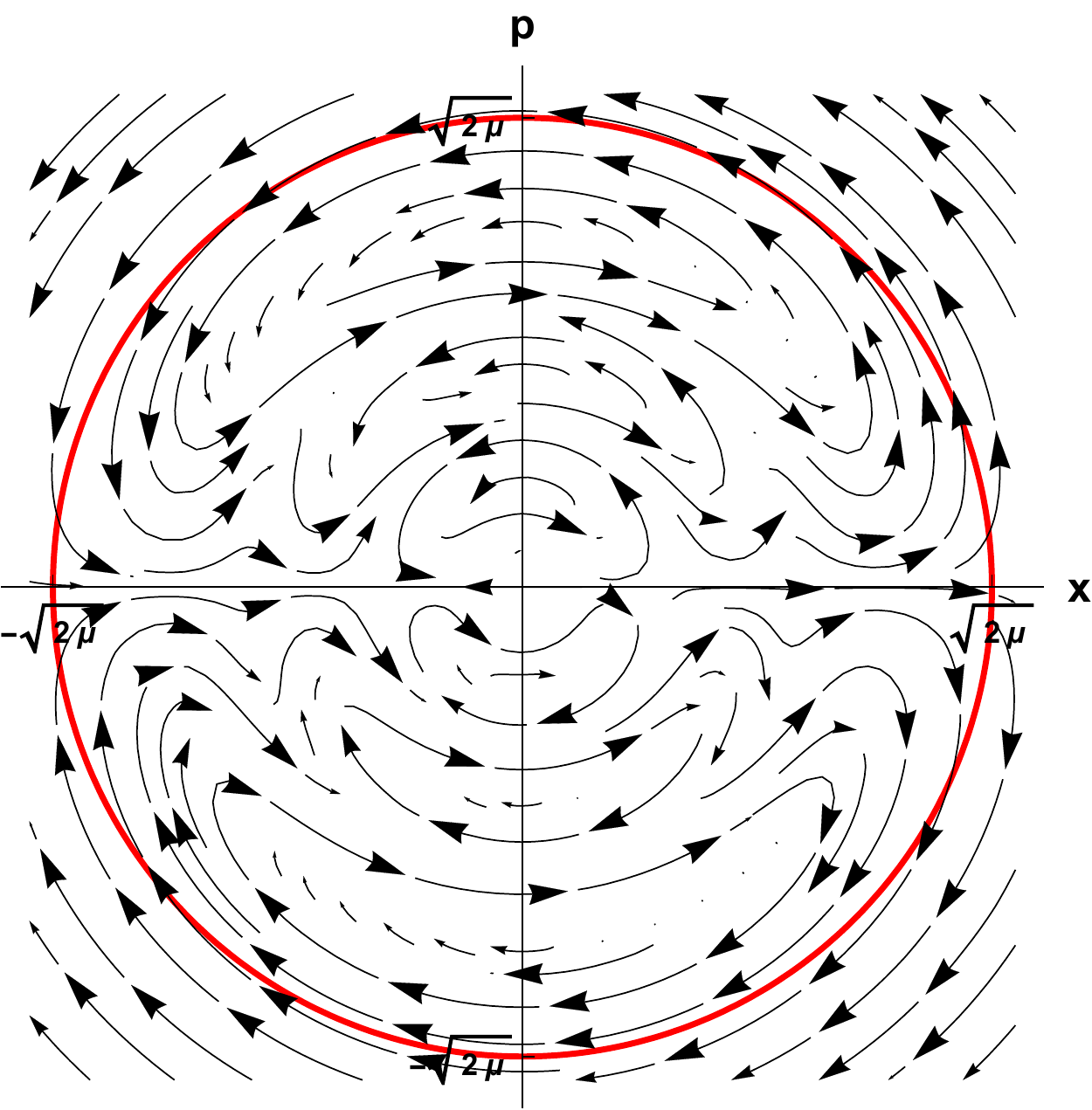}
		\caption{Trajectories generated by the Weyl symbol $\sigmaHc(x,p)$ with $N=7$.}
		\label{fig:1}
	\end{subfigure}\quad
	\begin{subfigure}[t]{0.31\columnwidth}
		\centering
		\includegraphics[width=\textwidth]{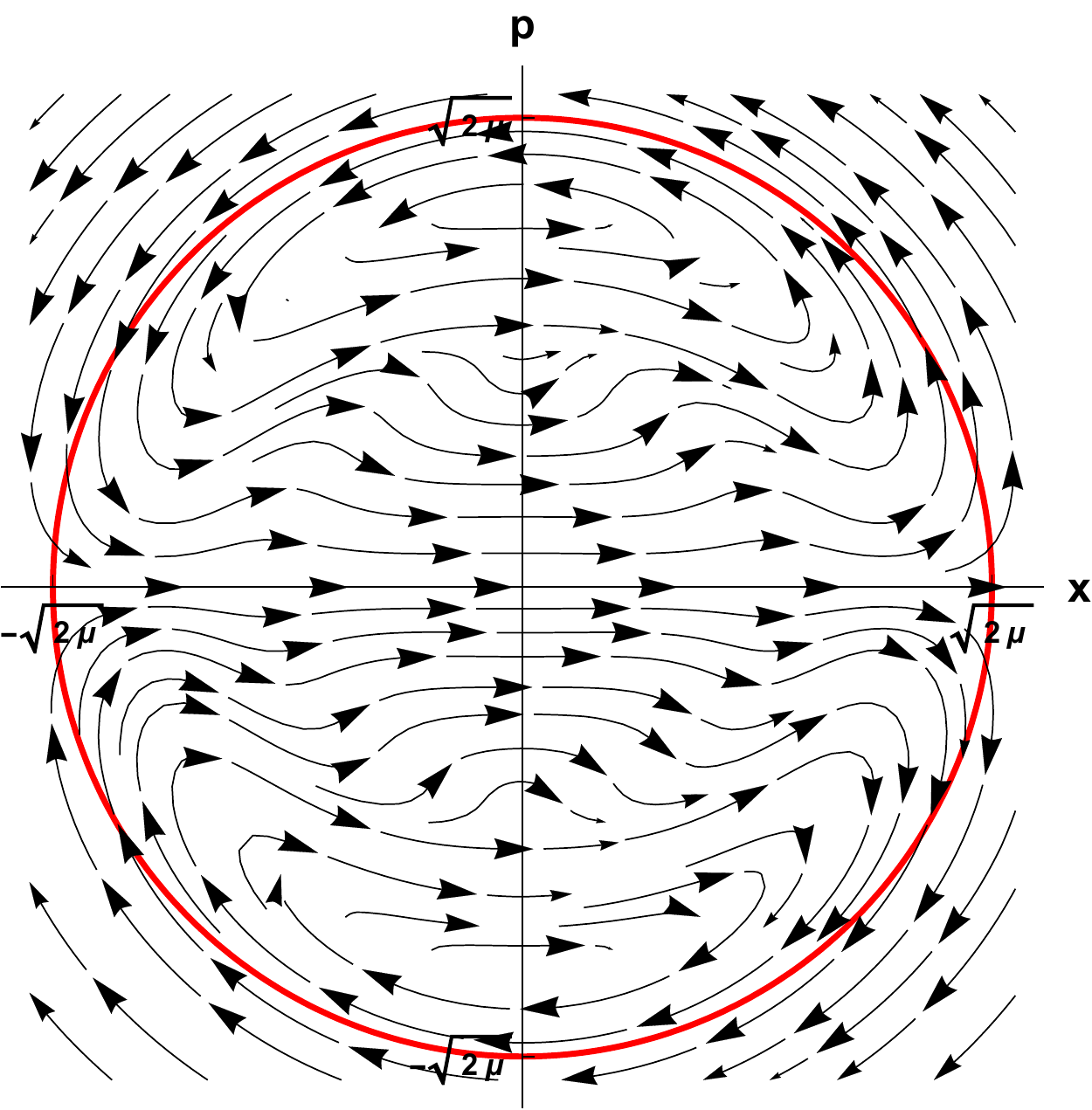}
		\caption{Trajectories generated by the smooth function $p\chi_D^{(N)}(x,p)$ with $N=7$. A particle near $\partial D$ moves at speed $\propto N^{\frac{2}{3}}$. }
		\label{fig:2}
	\end{subfigure}\quad 
		\begin{subfigure}[t]{0.31\columnwidth}
			\centering
		\includegraphics[width=\textwidth]{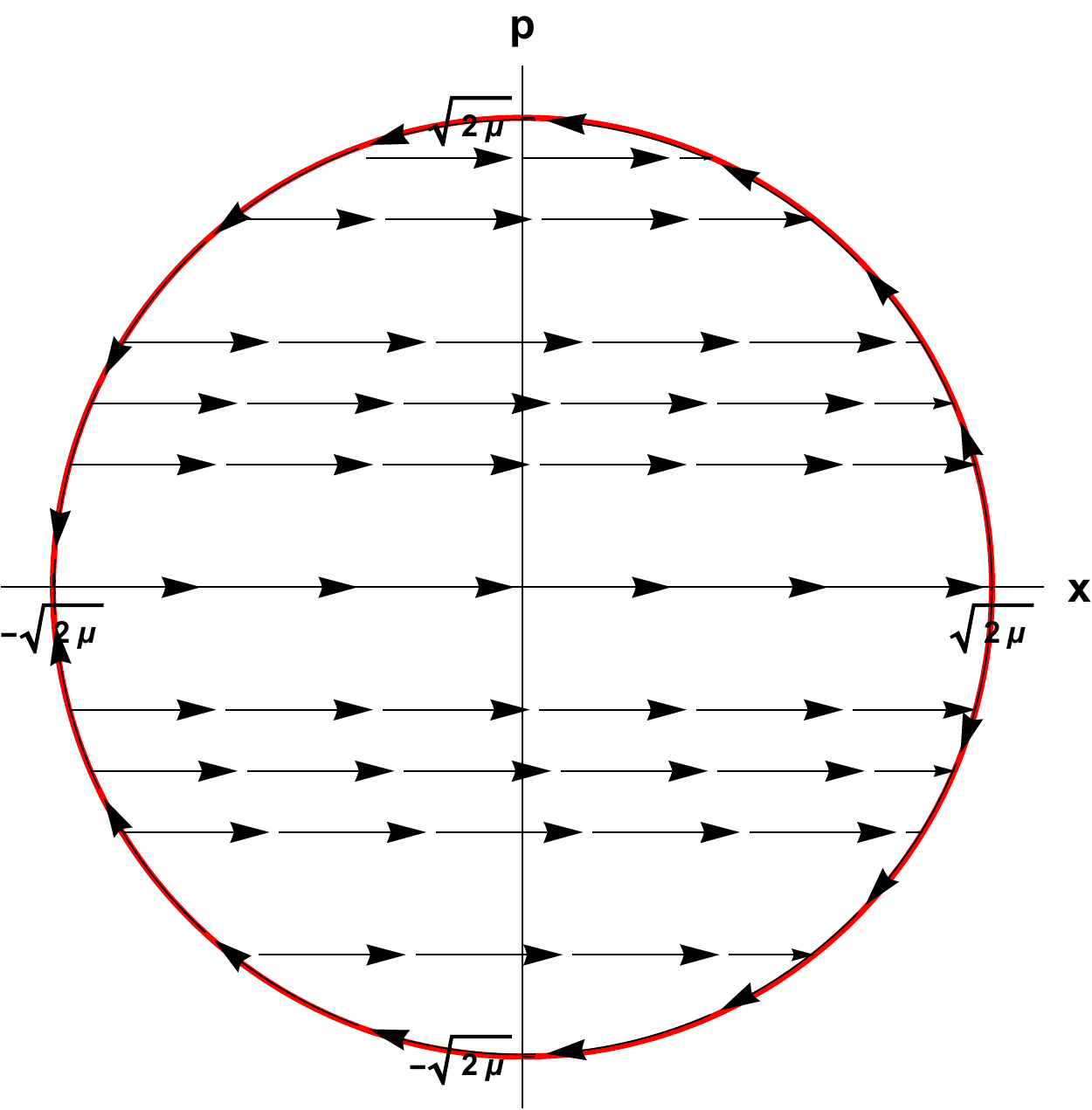}
		\caption{Trajectories generated by the discontinuous function $p\chi_D(x,p)$. The motion on $\partial D$ is at `infinite' speed. }
		\label{fig:3}
	\end{subfigure}
	\caption{Phase portraits for the Hamiltonian dynamics. The red solid line is the boundary $\partial D$ of the disk. Here $\mu=2$.} 
	\label{fig:trajectories}
\end{figure}

We know, however, by Theorem~\ref{thm:convergenceAiry_intro},  that the symbol $\sigmaQc(x,p)$ is asymptotically close to a smoothed version 
\be
p\chi_D^{(N)}(x,p)=p\Ai_1\left(\frac{(2N)^{\frac{2}{3}}}{2\mu}\left(r^2-2\mu\right)\right).
\ee 
For each $N$, it makes sense to consider the Hamiltonian system generated by  $p\chi_D^{(N)}(x,p)$, 
\be
\begin{cases}
\displaystyle\dot{x}=\dfrac{\partial }{\partial p}\left(p\chi_D^{(N)}(x,p)\right)\\
\displaystyle\dot{p}=-\dfrac{\partial}{\partial x}\left(p\chi_D^{(N)}(x,p)\right)
\end{cases} \qquad (\clubsuit).
\ee 
This is a family of well-defined Hamilton equations and we can expect, for large $N$, the solutions of $(\clubsuit)$ to be `close' to the sought semiclassical limiting dynamics.

We give here a sketch of an argument showing that for large $N$, the solutions of $(\clubsuit)$ behave as the formal solutions of  the singular problem $(\diamondsuit)$.
The equations of motions from $(\clubsuit)$ are
\begin{gather}
\displaystyle\dot{x}=\chi^{(N)}_{[0,\sqrt{2\mu})}(r)+p\delta^{(N)}_{\sqrt{2\mu}}(r)\dfrac{p}{r}, \qquad
\displaystyle\dot{p}=-p\delta^{(N)}_{\sqrt{2\mu}}(r)\dfrac{x}{r},
\end{gather}
where
\be
\chi^{(N)}_{[0,\sqrt{2\mu})}(r):=\Ai_1\left(\frac{(2N)^{\frac{2}{3}}}{2\mu}\left(r^2-2\mu\right)\right),\quad 
\delta^{(N)}_{\sqrt{2\mu}}(r):=-r\frac{(2N)^{\frac{2}{3}}}{\mu}\Ai\left(\frac{(2N)^{\frac{2}{3}}}{2\mu}\left(r^2-2\mu\right)\right).
\ee
Observe that $\chi^{(N)}_{[0,\sqrt{2\mu})}(r)$ are uniformly bounded functions that converge, as $N\to\infty$ to the characteristic function $\chi_{[0,\sqrt{2\mu})}$, see~\eqref{eq:asymp_Ai1}. The corresponding component of the field is of order $\mathrm{O}(1)$. The sequence of functions $\delta^{(N)}_{\sqrt{2\mu}}(r)$ converges to $\delta_{\sqrt{2\mu}}(r)$ in a distributional sense, as $N\to\infty$. This can be seen, in Fourier space, from the identity $
\int_{\R}\Ai(x)e^{-ikx}dx=e^{ik^3/3}$.  \par

We conclude that the Hamiltonian vector field generated by $p\chi_D^{(N)}(x,p)$ converges  to the singular vector field generated by $p\chi_D(x,p)$.  The  component of the field containing $\delta_{\sqrt{2\mu}}^{(N)}(r)$ is of order $\mathrm{O}(N^{\frac{2}{3}})$ and generates a motion at speed $\propto N^{\frac{2}{3}}$, which becomes `infinite' in the singular limit. Fig.~\ref{fig:trajectories} shows a comparison of the phase portraits of the Hamiltonian dynamics generated by the Weyl symbol $\sigmaHc(x,p)$, the smooth Hamiltonian $p\chi_D^{(N)}(x,p)$ and the discontinuous function $p\chi_D(x,p)$. Note the effective change of topology in the limit singular case that results from the instantaneous motion along the circle $\partial D$. 

\subsection{Spectral analysis of the Zeno Hamiltonian $H_N$}
The matrix representation of $H_N$ (in the Hermite basis $\left\{ \psi_k^{\hbar} \right\}_{k \in \N}$, see Appendix~\ref{app:A}) is the $N\times N$ complex Hermitian matrix
$$
H_N=i\sqrt{\frac{\hbar}{2}}
\begin{bmatrix}
0 & 1& 0&\cdots & 0\\
-1& 0& \sqrt{2} && \vdots \\
0 &-\sqrt{2} &\ddots &&\\ 
\vdots && &0&\sqrt{N-1} \\ 
0 &\cdots   &&-\sqrt{N-1} &0
\end{bmatrix}.
$$
This is a \emph{Jacobi matrix} about which we have very precise spectral information (characteristic polynomial, eigenvalues and their counting measure) for all $N$.
\begin{prop} For all $N \geq 1$,
\be
\det\left(yI_N-H_N\right)=\left(\frac{\sqrt{\hbar}}{2}\right)^Nh_{N}\left(\sqrt{\hbar}y\right), 
\ee
where $h_N$ is the \emph{Hermite polynomial} of degree $N$, see Appendix~\ref{app:A}. 
In particular, the eigenvalues of $H_N$ are the $N$ (simple and real) zeros of the Hermite function $\psi^{\hbar}_N(z)$.
\end{prop}
\begin{proof}
$H_N$ is unitarily equivalent to $P_{<N} \hat{x}P_{<N}$ (see equations~\eqref{eq:Hermite_x}-\eqref{eq:Hermite_p}), and so the two operators have equal characteristic polynomial. The claim now follows from a  result for general orthogonal polynomials on the real line due to Simon~\cite[Prop. 2.2]{Simon09}.  
\end{proof}
If $y^{(j)}_N$ are the zeros of $\psi^{\hbar}_N(y)$, we define the eigenvalues counting measure  $\nu_N$ of the Zeno Hamiltonian $H_N$ to be the nonnegative measure that puts weight $1/N$ on each eigenvalue of $H_N$ (the $y^{(j)}_N$'s). From well-known results on Hermite polynomials~\cite{Deift99} it follows that the measure $\nu_N$ weakly converges to the semicircular density in the simultaneous limit $\hbar\to0$, $N\to\infty$ with $\hbar N$ asymptotically fixed. See Figure~\ref{fig:semi}.
\begin{prop} \label{prop:semicircular} For all continuous bounded functions $f$, 
\be
\int_{\R} f(y)d\nu_N(y)\to\int_{\R}f(y)\rho_{\mu}(y)dy,
\ee
as $N\to\infty$, $\hbar\to0$, with the product $\hbar N$ converging to $\mu>0$.
\end{prop}
\begin{figure}[t]
	\centering
	\includegraphics[width=.55\textwidth]{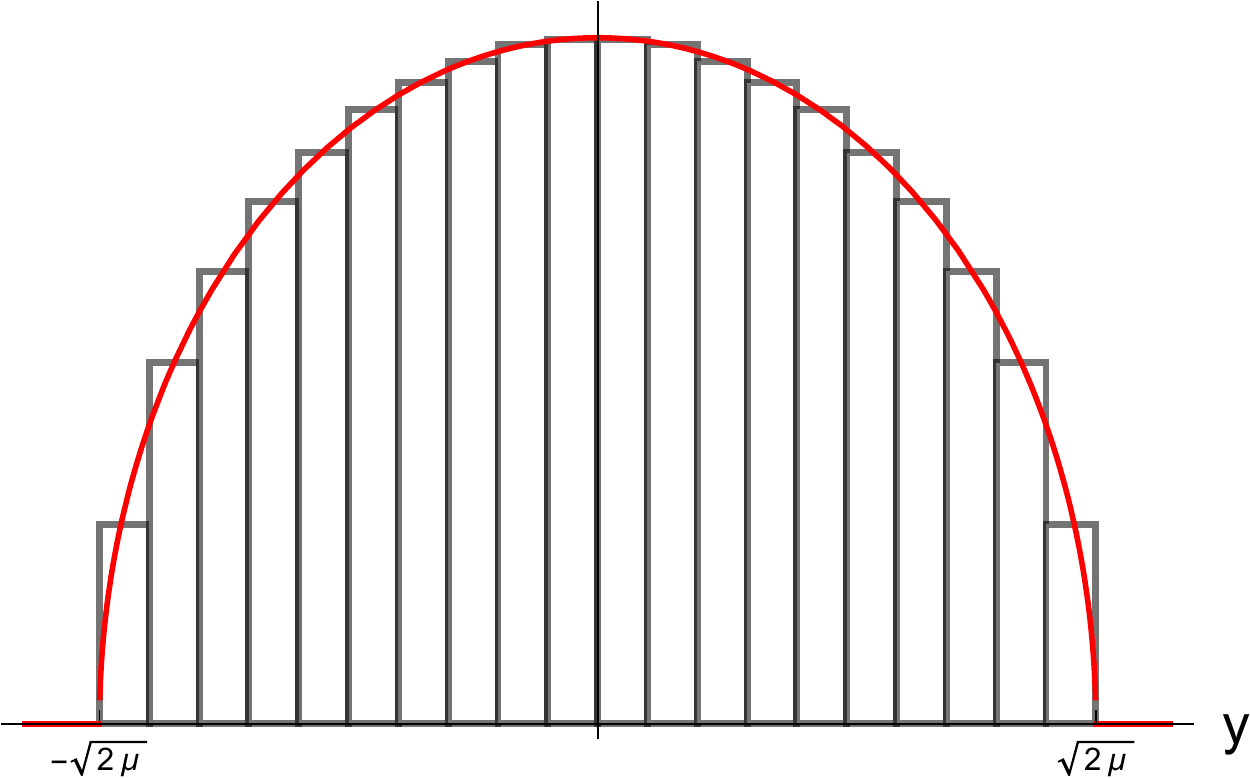}
	\caption{Illustration of Proposition~\ref{prop:semicircular}. The histogram of the eigenvalues of the Zeno Hamiltonian $H_N$ for $N=2000$, and $\hbar N=\mu=2$ is compared with the semicircular density $\rho_{\mu}(y)=\frac{1}{\pi\mu}\sqrt{(2\mu-y^2)_+}$ of Eq.~\eqref{eq:semic}.}
	\label{fig:semi}
\end{figure}
\begin{rem}
The semicircular spectral distribution can be obtained formally from the semiclassical limit of Theorem~\ref{thm:convergenceAprime_intro}. Indeed, in the limit symbol $p\chi_D(x,p)$ of the Zeno Hamiltonian is concentrated on the disk $D$ of radius $\sqrt{2\mu}$. Semiclassically, the density of the eigenvalues is the fraction of the phase space volume with energy between $y$ and $y+dy$:
$$
\frac{\mathrm{Area}\left(\{y\leq p\leq y+dy\}\cap D\right)}{\mathrm{Area}\left(D\right)}=\frac{2\sqrt{2\mu-y^2}dy}{\pi(2\mu)}=\rho_{\mu}(y).
$$
\end{rem}

\subsection{Proof strategy and relations to other works}
When $N$ is large, the symbols $\sigmaKc$ and $\sigmaQc$ are highly oscillating smooth functions. As discussed in Remark~\ref{rempoint}, looking for a global semiclassical limit in a pointwise sense is hopeless. It turns out that the sought convergence of the symbols holds in a weak sense if the set of test functions is chosen to be $\AA$.

The proofs presented in this paper are based on the following observations: 
\begin{enumerate}
\item The asymptotics of integrals of the Weyl symbols $\sigmaKc(x,p)$ and $\sigmaQc(x,p)$ against functions on the phase space is related to the pointwise asymptotics of the Fourier transforms $\FF_2\sigmaKc(x,y)$ and $\FF_2\sigmaHc(x,y)$;
\item The function $\FF_2\sigmaKc(x,y)$ is a sum of $N$ terms (cross products of Hermite functions), see
Eq.~\eqref{eq:K_N1}. However, thanks to the Christoffel-Darboux formula (Lemma~\ref{lem:CD}) this sum can be expressed in terms of the $N$-th and $(N-1)$-th Hermite functions only. Hence, studying the large $N$ asymptotics with $\hbar N=\mu$ amounts to study the large degree asymptotics of Hermite functions. This is a well-studied topic in the theory of orthogonal polynomials from which we can freely borrow explicit asymptotic formulae. So, to prove  the convergence of the symbols we will show the convergence of the Christoffel-Darboux kernel along with its derivatives to the \emph{sine} and \emph{Airy kernels} (in the formulation presented in the book of Anderson, Guionnet and Zeitouni~\cite{Anderson10})
\item The symbol $\sigmaQc(x,p)$ is `asymptotically close' to $p\sigmaKc(x,p)$ in the dual space $\AA'$ (Proposition~\ref{prop:bound_symbols}). This is suggested by the heuristic observation that, in  the limit $\hbar\to0$, the algebra of observables should become commutative. What we gain is that, once we know the asymptotics of $\sigmaPc(x,p)$ we can directly deduce the asymptotics of  $\sigmaQc(x,p)$.
\end{enumerate}

The seminal paper by Lions and Paul~\cite{Lions93} on the semiclassical limit of Wigner measures, and the more recent developments~\cite{Ambrosio11,Athanassoulis13,Curtright01,Figalli12}  were  instrumental in our study.

We mention that the symbol $\sigmaPc(x,p)$ of the orthogonal projection $P_{<N}$  studied in the present paper has close connection to the fuzzy approximation of two-dimensional disk proposed by Lizzi, Vitale and Zampini~\cite{Lizzi03,Lizzi06}. A \emph{fuzzy space} is an approximation of an abelian algebra of functions on an ordinary space with a sequence of finite-rank matrix algebras, which preserve the symmetries of the original space, at the price of non-commutativity. Eq.~\eqref{eq:asymp_sigma_P_<N} of Theorem~\ref{thm:convergenceAprime_intro} is the precise mathematical statement behind the numerical results of~\cite{Lizzi03,Lizzi06}. To our knowledge, the finer asymptotics of Theorem~\ref{thm:convergenceAiry_intro} is a new result, that has not been observed numerically neither.

The convergence of symbols of projection operators to the characteristic function of the classically allowed region is folklore in theoretical physics. In recent years, there has been an explosion of results on the asymptotics of the Christoffel-Darboux kernel for orthogonal polynomials on the real line, especially in connections to eigenvalue statistics of random matrices and integrable probability models~\cite{Deift99,Mehta04,Peres06,Romik15}. The interest to these asymptotics in the theoretical and mathematical physics community has been mostly motivated by applications to the number statistics of non-interacting fermions. The asymptotics at the `edge' has been also investigated at various levels of rigour.  See, e.g.~\cite{Bettelheim11,Bettelheim12,Bornemann16,Cunden18,Cunden19,Dean18,Dean19,Deleporte21,Eisler13,Torquato08}. 

The semiclassical structure of quite general \emph{cut quantum observables} $\Pi Q\Pi$ (with $\Pi$ a spectral projection and $Q$ a pseudodifferential operator) was  studied by Hernandez-Duenas and Uribe~\cite{Hernandez-Duenas15}. Their results is consistent with Theorem~\ref{thm:convergenceAprime_intro} of the present paper. Those authors also studied the unitary dynamics generated by the cut quantum observables (the analogue of the Zeno Hamiltonian $H_N$ of the present paper), and numerically found fascinating phenomena of splitting of the wave-packets and infinite propagation speed near the boundary of the classically allowed region. This is again consistent with our findings.

Given the universality results on the asymptotics of orthogonal polynomials and random matrices~\cite{Deift99}, we expect that Theorem~\ref{thm:convergenceAprime_intro} is valid for a rather large class of symbols associated to finite-rank orthogonal projections. The recent paper by Deleporte and Lambert~\cite{Deleporte21} suggests that Theorem~\ref{thm:convergenceAiry_intro} would be valid as long as the gradient of the confining potential does not vanish at the points of classical inversion of motion.  In any case the statement of analogues of Theorem~\ref{thm:convergenceAiry_intro} should depend on the geometry of the level sets of the corresponding classical Hamiltonian function. Further study is in progress.

\subsection{Outline of the paper}
The structure of the paper is as follows. In the next section we recall some preliminary background material, introduce a precise presentation of the model and provide the calculation of the symbols. In Section~\ref{sect:scaling} we discuss the different scaling limits in Theorems~\ref{thm:convergenceAprime_intro} and~\ref{thm:convergenceAiry_intro}.
Section~\ref{sec:proof} is entirely devoted to the proofs of the main technical results, and of Theorems~\ref{thm:convergenceAprime_intro}, \ref{thm:convergenceAiry_intro} and ~\ref{thm:edge_out_pointwise}. The paper includes two appendices. In Appendix~\ref{app:A} we collect known formulae on the Hermite functions. Appendix~\ref{app:B} contains the definition and a few properties of the sine and the Airy kernel.

 \section{Notation, preliminaries, Weyl symbols and kernels} \label{sec:notation}
We first introduce some notation and preliminary notions that we use throughout this work. For a linear  operator $L$ on $L^2(\R)$ we write  $L\doteq L(u,v)$ to indicate that $L$ has kernel $L(u,v)\in L^2(\R\times\R)$. In this paper all kernels are continuous. 
For $A,B$ linear operators, we use the notation $[A,B]:=AB-BA$ for the commutator. 
Let $D$ be the linear operator defined, for  $f\in H^1(\mathbb{R})$, by the formula $(Df)(u)=\frac{d}{du}f(u)$. We have
\be
[D,L]\doteq\left(\frac{\partial}{\partial u}+\frac{\partial}{\partial v}\right)L(u,v).
\label{eq:commutator_DL}
\ee
For $x\in\R$ and $\gamma>0$, let
  \begin{alignat}{2}
V_{x,\gamma}\colon L^2(\R)&\longrightarrow&& L^2(\R) \nonumber \\
f(u)&\longmapsto&& \left(V_{x,\gamma}f\right)(u):=\sqrt{\gamma}f\left(x+\gamma u\right). 
\label{eq:unitary_affine}
\end{alignat}
Of course $\left(V^{-1}_{x,\gamma}f\right)(u)=\sqrt{1/\gamma}f\left(\gamma^{-1}(u-x)\right)$, and $V_{x,\gamma}$ is unitary. If we conjugate the operator $L$ by the scaling unitary $V_{x,\gamma}$, its  kernel gets changed into  
\begin{align}
V_{x,\gamma}LV_{x,\gamma}^{-1}
\doteq \gamma L\left(x+\gamma u,x+\gamma v\right).
\end{align}

We shall consider the following space of test functions introduced by Lions and Paul~\cite{Lions93},
\be
\AA=\left\{f\in C^{}_0(\R_x\times\R_p)\colon  \|f\|_{\AA}:=\int_{\R} dy\sup_x\left|\FF_2 f(x,y)\right| < \infty \right\},
\ee
where $C^{}_0(\R_x\times\R_p)$ is the usual space of continuous functions tending to zero at infinity.

$\AA$ is a Banach algebra with the following properties (see~\cite{Lions93}): 
\begin{itemize}
\item[-] $\mathcal{S}(\R_x\times\R_p)$, $C^{\infty}_c(\R_x\times\R_p)$, and $\BB=\{f\in\AA\colon \FF_2f\in C^{}_c(\R_x\times\R_y)\}$ are dense subspaces in $\AA$.
\item[-] $\sup_{x,p}|f(x,p)|\leq (1/2\pi)\|f\|_{\AA}$; hence $\AA$ is contained in the space of bounded continuous functions in the phase space $C_b(\R_x\times\R_p)$.
\end{itemize}
Let $\AA'$ be the dual of $\AA$. From the Parseval identity
it follows that
\be
\|h\|_{\AA'}=\frac{1}{2\pi}\sup_y\int \left|\FF_2 h(x,y)\right|dx.
\ee
A basic property is $\|h\|_{\AA'}\leq \frac{1}{2\pi} \|h\|_{L^1}$ (hence $L^1(\R_x\times\R_p)\subset \AA'$).

\begin{definition}\label{def:Weyl}
Given a number $\hbar>0$, the Weyl symbol of the operator $L\doteq L(u,v)$ is defined as
\be
\sigma_L^{\hbar}(x,p):=\int_{\R_y}\hbar L\left(x-\frac{\hbar y}{2},x+\frac{\hbar y}{2}\right)e^{ipy}dy.
\ee
Equivalently, $\sigma_L^{\hbar}$ is defined by the identity
\be
\left(\FF_2 \sigma_L^{\hbar}\right)(x,y)=\left(2\pi\hbar\right) L\left(x-\frac{\hbar y}{2},x+\frac{\hbar y}{2}\right)
\ee
and, by Plancherel's theorem,
\be
\int\limits_{\R_x\times\R_p}\overline{\sigma_L^{\hbar}(x,p)}\varphi(x,p) dx dp=\int\limits_{\R_x\times\R_y}\overline{\hbar L\left(x-\frac{\hbar y}{2},x+\frac{\hbar y}{2}\right)}\FF_2\varphi(x,y)dxdy.
\label{eq:Plancherel_f}
\ee
\end{definition}
We will often use the shorthand 
\be
\langle \sigma,f\rangle:=\int_{\R_x\times\R_p} \overline{\sigma(x,p)}f(x,p) dx dp.
\ee

We recall that the Weyl symbol of the product of two operators $A$, $B$ is \emph{not} the ordinary product of the symbols $\sigma^{\hbar}_{AB}\neq \sigma^{\hbar}_{A}\sigma^{\hbar}_{B}$, unless $A$ and $B$ commute. 
The noncommutative Moyal product  $\sharp$ is defined as the composition law that does the job: $\sigma^{\hbar}_{AB}= \sigma^{\hbar}_{A}\, \sharp \, \sigma^{\hbar}_{B}$~\cite{folland}.
\begin{definition}
Given two linear operators $A$ and $B$ on $L^2(\R)$ with Weyl symbols $\sigma_A^\hbar$ and $\sigma_B^\hbar$ respectively, the Moyal product is defined as follows:
$$
\sigma^{\hbar}_{A} \, \sharp \, \sigma^{\hbar}_{B}(x,p)=\int_{\R^4} \sigma^{\hbar}_{A}(x_1,p_1)\sigma^{\hbar}_{B}(x_2,p_2) e^{\frac{2i}{\hbar}[(x-x_1)(p-p_2)-(x-x_2)(p-p_1)]}\frac{dx_1 dp_1 dx_2 dp_2 }{(\pi \hbar)^2}.
$$
\end{definition}

Recall that the normalised eigenfunctions of the harmonic oscillator operator $H_{\operatorname{h.o.}}$ in~\eqref{eq:Hho} are the Hermite functions 
\be
\psi_k^{\hbar} (x)=\sqrt{\frac{\alpha }{\sqrt{\pi } 2^k k!}} \exp \left(-\frac{1}{2}\alpha ^2 x^2\right) h_k(\alpha  x),\qquad k=0,1,2,\ldots
\ee
where $\alpha^2=1/\hbar$ and
\be
h_k(y)=(-1)^ke^{y^2}\frac{d^k}{dy^k}e^{-y^2}
\ee
is the $k$-th \emph{Hermite polynomials}, see Appendix~\ref{app:A}. Consider the \emph{orthogonal projection} 
$$P_{<N}=\chi_{(-\infty,\hbar N)}(H_{\operatorname{h.o.}})$$ onto the span of the first  $N$ Hermite eigenfunctions in~(\ref{eqn:PN}).  The  \emph{Zeno Hamiltonian} in~(\ref{eqn:HN}) is the \emph{truncated momentum operator} 
\be
H_N=P_{<N}\hat{p}P_{<N}=\hat{p}P_{<N}-[\hat{p},P_{<N}]P_{<N}.
\ee
\begin{prop}[Integral kernels]
\be
P_{<N}\doteq K_N^{}(u,v)=\sum_{k=0}^{N-1}\psi_k^{\hbar} (u)\psi_{k}^{\hbar} (v)
\label{eq:K_N1}
\ee
\begin{align}\label{kernelHN}
P_{<N}\hat{p}P_{<N}\doteq Q_N^{}(u,v)&=\int_{\R}K_N(u,w)\left(-i\hbar\frac{\partial}{\partial w}\right)K_N(w,v)dw\nonumber\\
&=i \sqrt{\frac{\hbar }{2}}\sum _{j=0}^{N-2} \sqrt{j+1} \left[\psi^{\hbar}_{j+1} \left(u\right) \psi^{\hbar}_j \left(v\right)-\psi^{\hbar}_j \left(u\right) \psi^{\hbar}_{j+1} \left(v\right)\right],
\end{align}
\be
\label{eq:commutator_kernel2}
[\hat{p},P_{<N}]P_{<N}\doteq R_N(u,v)= i\sqrt{\frac{\hbar N}{2}}\psi^{\hbar}_{N} \left(u\right) \psi^{\hbar}_{N-1}\left(v\right).
\ee
\end{prop}
\begin{proof}
Formula \eqref{eq:K_N1} follows directly by the definition of the Hermite functions.  Formula \eqref{kernelHN} is obtained by a direct calculation using the three-term recurrence~\eqref{eq:Hermite_p}, while~(\ref{eq:commutator_kernel2}) follows by applying the identity~\eqref{eq:commutator_DL} to $P_{<N}$, and using the orthonormality of the  eigenfunctions  $\{  \psi^{\hbar}_{k}\}_{k \in \N}$.
\end{proof}
\begin{rem}
The kernels $K_N$, $Q_N$ and $R_N$ are rapidly decreasing functions in  $\SS(\R_u\times \R_v)$.
\end{rem}
The Weyl symbols of $P_{<N}$ and $H_N$ 
\begin{align}
\sigmaPc(x,p)=
\int_{\R}\hbar K_N\left(x-\frac{\hbar y}{2},x+\frac{\hbar y}{2}\right)e^{ipy}dy,
\\
\sigmaHc(x,p)=
\int_{\R}\hbar Q_N\left(x-\frac{\hbar y}{2},x+\frac{\hbar y}{2}\right)e^{ipy}dy,
\end{align}
have explicit representations in terms of associated Laguerre polynomials (this is a manifestation of the  so-called `Laguerre connection'~\cite[\S 1.9]{folland}).
\begin{prop}[Weyl symbols]
\label{prop:symbols_finiteN} For all $x,p \in \R$: 
\begin{align}
\sigmaPc(x,p)&=2e^{-(p^2+x^2)/\hbar}\sum_{j=0}^{N-1}(-1)^jL_j\left(2(p^2+x^2)/\hbar\right),\\
\sigmaHc(x,p)&=4pe^{-(p^2+x^2)/\hbar}\sum_{j=0}^{N-2}(-1)^{j}L_j^{(1)}\left(2(p^2+x^2)/\hbar\right),
\end{align}
where 
$$
L_k^{(j)}(y)=\sum_{m=0}^k \frac{(k+j)!}{(k-m)!(j+m)!m!} (-y)^m
$$
are the \emph{associated Laguerre polynomials}.
\end{prop}
\begin{proof}
A consequence of the following formula by Groenewold~\cite{Groenewold46} valid for all $j\leq k$ (we write the formula as in~\cite[Eq. (30)]{Curtright01}),
\begin{multline}
\int\psi_j^{\hbar}\left(x-\frac{y}{2}\right)\psi_k^{\hbar}\left(x+\frac{y}{2}\right)e^{ipy}dy\\=
2\sqrt{\left(\frac{2}{\hbar}\right)^{k-j}\frac{j!}{k!}}\left( x + i p\right)^{k-j}e^{-(p^2+x^2)/\hbar}(-1)^jL_j^{(k-j)}\left( 2(p^2+x^2)/\hbar\right).
\end{multline}

\end{proof}

\begin{rem} 
\label{rem:complex}
The symbols $\sigmaPc$ and $\sigmaHc$ are rapidly decreasing functions in $\SS(\R_x\times \R_p)$. Notice that $\sigmaPc$ is rotational symmetric.
 It may be convenient in the following to consider  $\sigmaPc$ and $\sigmaHc$ as complex-valued functions defined on the complexification $\C_x\times\C_p$ of the real phase space. They are entire functions in both variables $x$ and $p$.
\end{rem}

\section{Scaling limits}\label{sect:scaling}
In this section we provide an heuristic explanation of the different scaling limits in Theorems~\ref{thm:convergenceAprime_intro} and~\ref{thm:convergenceAiry_intro}. The following discussion is somewhat breezy. For a more careful exposition of similar ideas, see~\cite{Bornemann16,Cunden18,Deleporte21}.

Recall that $P_{<N} \doteq K_N(u,v)$ and 
\be
\label{eq:symb_ker}
\FF_2\sigmaKc(x,y)=2\pi\hbar K_N\left(x-\frac{\hbar y}{2},x+\frac{\hbar y}{2}\right).
\ee
At scale $\hbar$, the kernel has an asymptotic limit that can be identified as follows.  
We start by writing the rescaled  kernel in terms of the conjugation of a unitary transformation on the operator. 
If we conjugate the projection $P_{<N}$ by the scaling unitary $V_{x,\hbar}$ in~\eqref{eq:unitary_affine}, we get that the kernel of the rescaled projection is the rescaled kernel:
\begin{align}
V_{x,\hbar}P_{<N}V_{x,\hbar}^{-1}=\chi_{(-\infty,\hbar N)}\left(V_{x,\hbar}H_{\operatorname{h.o.}}V_{x,\hbar}^{-1}\right)
\doteq \hbar K_N\left(x+\hbar u,x+\hbar v\right).
\end{align}
The action of the rescaled harmonic oscillator operator on a function $f$ in its domain is
\begin{align*}
\left(V_{x,\hbar}H_{\operatorname{h.o.}}V_{x,\hbar}^{-1}f\right)(u)=\frac{1}{2}\left[-\frac{d^2 }{d u^2}+\hbar^2 u^2 +2\hbar u x +x^2\right]f(u).
\end{align*}
So we expect that 
\be
\chi_{(-\infty,\hbar N)}\left(V_{x,\hbar}H_{\operatorname{h.o.}}V_{x,\hbar}^{-1}\right)
\simeq  \chi_{(-\infty,2\mu-x^2)}\left(-\frac{d^2 }{d u^2}\right),\quad \text{ for $\hbar\to0$, $N\to\infty$, with $\hbar N=\mu$}.
\label{eq:approx_proj1}
\ee
We recall the following result adapted from~\cite[Lemma A.5]{Bornemann16}.
\begin{lem}
\label{lem:Greenf1} 
The operator $-\frac{d^2 }{d u^2}$ is essentially self-adjoint on $C^{\infty}_c(\R)$, and its unique self-adjoint extension has only absolutely continuous spectrum $\sigma\left(-\frac{d^2 }{d u^2}\right)=\sigma_{\mathrm{ac}}\left(-\frac{d^2 }{d u^2}\right)=[0,\infty)$. Moreover, 
\be
\chi_{(-\infty,2\mu-x^2)}\left(-\frac{d^2 }{d u^2}\right)\doteq 
\mu \rho_\mu (x)K_{\sine}\left(\mu \rho_\mu (x)u,\mu \rho_\mu (x)v\right),
\ee
where $K_{\sine}$ is the sine kernel~\eqref{eq:sinek}.
\end{lem}
From~\eqref{eq:symb_ker}, we see that a rescaling $\hbar$ in the Hilbert space $L^2(\R)$ corresponds to zooming at scale $\hbar^0$ in the phase space. The precise statement of~\eqref{eq:approx_proj1} is Proposition~\ref{prop:Asymptotics_K_N} in Section~\ref{sec:proof}.

To explain how a different asymptotics arises at the boundary $\partial D$, we need to study the rescaled harmonic oscillator operator in a neighbourhood of the classical turning points $x=\pm\sqrt{2\mu}$. Let us zoom at scale $\hbar^\alpha$, with $\alpha>0$ an exponent to be determined:
\begin{align*}
\left(V_{\sqrt{2\mu},\hbar^{\alpha}}H_{\operatorname{h.o.}}V_{\sqrt{2\mu},\hbar^{\alpha}}^{-1}f\right)(u)=\frac{1}{2}\left[-\hbar^{2(1-\alpha)}\frac{d^2}{d u^2}+\hbar^{2\alpha} u^2+2^{\frac{3}{2}}\mu^{\frac{1}{2}}\hbar^\alpha u +2\mu\right]f(u).
\end{align*}
If we choose $\alpha=\frac{2}{3}$ we then expect that 
\be
\chi_{(-\infty,\hbar N)}\left(V_{\sqrt{2\mu},\hbar^{\frac{2}{3}}}H_{\operatorname{h.o.}}V_{\sqrt{2\mu},\hbar^{\frac{2}{3}}}^{-1}\right)
\simeq  \chi_{(-\infty,0)}\left(-\frac{d^2}{d u^2}+c_{\mu}^3\hat{u}\right),
\label{eq:approx_proj2}
\ee
for $\hbar\to0$, $N\to\infty$, with $\hbar N=\mu$, where $\hat{u}$ is the position operator and $c_{\mu}=2^{\frac{1}{2}}\mu^{\frac{1}{6}}$ is a constant given in~\eqref{eq:const}. Thus, the limit  at the edge is related to the Airy differential operator for which we have the following spectral result, adapted from~\cite[Lemma A.7]{Bornemann16}.
\begin{lem}
\label{lem:Greenf2} 
The operator $-\frac{d^2}{d u^2}+c_{\mu}^3\hat{u}$ is essentially self-adjoint on $C^{\infty}_c(\R)$, and its self-adjoint extension has only absolutely continuous spectrum $\sigma\left(-\frac{d^2}{d u^2}+c_{\mu}^3\hat{u} \right)=\sigma_{\mathrm{ac}}\left(-\frac{d^2}{d u^2}+c_{\mu}^3\hat{u}\right)=(-\infty,\infty)$. Moreover,
\be
\chi_{(-\infty,0)}\left(-\frac{d^2}{d u^2}+c_{\mu}^3\hat{u}\right)\doteq 
c_{\mu}K_{\Ai}(c_{\mu}u,c_{\mu}v).
\ee
where $K_{\Ai}$ is the Airy kernel~\eqref{eq:Airyk}.
\end{lem}
The precise statement of~\eqref{eq:approx_proj2} is Proposition~\ref{prop:Asymptotics_K_N2} in Section~\ref{sec:proof}.

From~\eqref{eq:symb_ker}, we see that a rescaling $\hbar^{\frac{2}{3}}$ at the edge in the Hilbert space $L^2(\R)$ corresponds to zooming at scale $\hbar^{\frac{2}{3}-1}=\hbar^{-\frac{1}{3}}$ around the boundary $\partial D$ in the phase space. This explains the rescaling in Theorem~\ref{thm:convergenceAiry_intro}.

\section{Proofs}
\label{sec:proof}
The proofs presented in this section are based on the following three observations: 
\begin{enumerate}
\item The asymptotics of the Weyl symbols $\sigmaKc$, $\sigmaQc$ is related (by Fourier transform in the second variable $\FF_2$) to the asymptotics of the integral kernels ${K_N}(x,y)$ and ${Q_N}(x,y)$.
\item The kernel ${K_N}(x,y)$ is a sum of $N$ terms (cross products of Hermite functions), see
Eq.~\eqref{eq:K_N1}. However, thanks to Christoffel-Darboux formula this sum can be expressed in terms of the $N$-th and $(N-1)$-th Hermite functions only. Hence, studying the large $N$ asymptotics with $\hbar N\sim\mu$ amounts to study the large degree asymptotics of the Hermite functions. 
\item  $\sigmaQc(x,p)$ is `asymptotically close' to $p\sigmaKc(x,p)$ for $N \to \infty$, $\hbar \to 0$ with $\hbar N = \mu >0$, see Proposition~\ref{prop:bound_symbols}, therefore
once we know the asymptotics of $K_N$ (and hence of $\sigmaKc$) we can directly deduce the asymptotics of  $\sigmaQc$.
\end{enumerate}
\subsection{Asymptotics of the kernels}
%\label{sub:asymp_kernels}
By telescoping the sum in~\eqref{eq:K_N1} and using the three-term relation~\eqref{eq:Hermite_x}, we get the celebrated \emph{Christoffel-Darboux formula}~\cite{Simon08}.
\begin{lem}\label{lem:CD} For all $u,v\in\mathbb{R}$, 
\be
K_N(u,v)=
\begin{cases}
\displaystyle\sqrt{\frac{\hbar N}{2}}\frac{\psi_{N}^{\hbar} (u)\psi_{N-1}^{\hbar} (v)-\psi_{N-1}^{\hbar} (u)\psi_{N}^{\hbar} (v)}{u-v}&\text{if $u\neq v$}\\\\
\displaystyle\sqrt{\frac{\hbar N}{2}}({\psi_{N}^{\hbar}}' (u)\psi_{N-1}^{\hbar} (u)-\psi_{N}^{\hbar} (u){\psi_{N-1}^{\hbar}}'(u))&\text{if $u= v$}
\end{cases}.
\ee
\end{lem}
Thus, the large-$N$ asymptotics of the \emph{Christoffel-Darboux kernel} $K_N(x,y)$ boils down to the classical subject of large degree asymptotics of orthogonal polynomials. 
A consequence of the Plancherel-Rotach asymptotics for $\psi_N^{\hbar}(x)$ (Equations~\eqref{eq:PR_bulk}-\eqref{eq:PR_edge}) are the following asymptotic behaviours of the kernel $K_N(x,y)$. 
\begin{prop}[Bulk asymptotics of the Christoffel-Darboux kernel]
\label{prop:Asymptotics_K_N}
Suppose that $\hbar=\hbar_N$ is the sequence defined by the condition $\hbar N=\mu$. Then, for any compact sets $U\Subset\R$ and $V\Subset \R^2$, and for any  $\alpha,\beta\in\{0,1\}$, there exists a constant $C>0$ such that
\be
\sup_{x\in U}\sup_{ (t,s)\in V}\left|\partial^{\alpha}_t\partial^{\beta}_s\left\{\hbar K_N\left(x+\hbar t,x+\hbar s\right)-\mu\rho_{\mu}(x)K_{\sine}\left(\mu\rho_{\mu}(x)t,\mu\rho_{\mu}(x)s\right)\right\}\right|\leq C \hbar,
\label{eq:convergence_K_N_K_sine}
\ee
where  $\rho_{\mu}(x)$ is the semicircular density~\eqref{eq:semic}, and $K_{\sine}$ is the sine kernel~\eqref{eq:sinek}. 
\end{prop}
\begin{prop}[Edge asymptotics of the Christoffel-Darboux kernel]
\label{prop:Asymptotics_K_N2}
Suppose that $\hbar=\hbar_N$ is the sequence defined by the condition $\hbar N=\mu$.  For any compact set $W\Subset \C^2$, and for any  $\alpha,\beta\in\{0,1\}$, there exists a constant $C>0$ such that
    \be
  \sup_{ (t,s)\in W}\left| \partial^{\alpha}_t\partial^{\beta}_s\left\{ \hbar^{\frac{2}{3}}K_N\left(\sqrt{2\mu}+\hbar^{\frac{2}{3}}t,\sqrt{2\mu}+\hbar^{\frac{2}{3}}s\right)
    -c_{\mu}K_{\Ai}\left(c_{\mu}t,c_{\mu}s\right)\right\}\right|\leq C\hbar^{\frac{1}{3}},
    \ee
    where $c_{\mu}$ is given in~\eqref{eq:const}, and  $K_{\Ai}$ is the Airy kernel~\eqref{eq:Airyk}.
\end{prop}
The scaling limits of the Christoffel-Darboux kernel to the sine and Airy kernel are well-known results. It is perhaps less known that the local uniform convergence can be promoted to their derivatives as well. We outline here a proof, adapting the presentation of the book by Anderson, Guionnet and Zeitouni~\cite[Chap. 3]{Anderson10}.

\begin{notation}
From now on, $(\hbar_N)_{N\geq1}$ is the positive sequence such that product $\hbar_NN=\mu$, where $\mu$ is a fixed positive number. We will write $\hbar$ instead of $\hbar_N$ for short, when no confusion arises.
We will also use the following shorthand
\begin{gather}
K_{N,x_0,\gamma}(t,s):=\gamma K_{N}(x_0+\gamma t,x_0+\gamma s).
\end{gather}
\end{notation}
\begin{proof}[Proof of Propositions~\ref{prop:Asymptotics_K_N} and~\ref{prop:Asymptotics_K_N2}] 
Consider first the case $\alpha=\beta=0$:
\be
\label{eq:proof1}
\sup_{x\in U}\sup_{ (t,s)\in V}\left| K_{N,x,\hbar}\left(t,s\right)-\mu\rho_{\mu}(x)K_{\sine}\left(\mu\rho_{\mu}(x)t,\mu\rho_{\mu}(x)s\right)\right|\leq C \hbar,
\ee
and 
\be
\label{eq:proof2}
  \sup_{ (t,s)\in W}\left| K_{N,\sqrt{2\mu},\hbar^{\frac{2}{3}}}\left(t,s\right)
    -c_{\mu}K_{\Ai}\left(c_{\mu}t,c_{\mu}s\right)\right|\leq C\hbar^{\frac{1}{3}}.
\ee
It is useful to get rid of the removable singularity $t=s$ in $K_{N,x,\hbar}$.
Toward this end, noting that for any differentiable functions $f, g$ on $\R$,
$$\frac{f (t)g(s) - f (s)g(t)}{
t-s}=
   g(s)\int_{0}^{1}
f' (\lambda t+(1-\lambda)s)d\lambda-
   f(s)\int_{0}^{1}
g' (\lambda t+(1-\lambda)s)d\lambda,
$$
we deduce that
\begin{align*}
K_{N,x,\hbar}(t,s)
&=
\sqrt{\frac{\hbar N}{2}}\psi_{N-1}^{\hbar} (x+\hbar s)\int_0^1{\psi_{N}^{\hbar}}'(\lambda (x+\hbar t)+(1-\lambda)(x+\hbar s))d\lambda\\
&-\sqrt{\frac{\hbar N}{2}}\psi_{N}^{\hbar} (x+\hbar s)\int_0^1{{\psi}_{N-1}^{\hbar}}'(\lambda (x+\hbar t)+(1-\lambda)(x+\hbar s))d\lambda\\
&=
\sqrt{\frac{\hbar N}{2}}\psi_{N-1}^{\hbar} (x+\hbar s)\int_0^1
\left(\sqrt{\frac{2N}{\hbar}}\psi_{N-1}^{\hbar} (z)-\frac{z}{\hbar}\psi_{N}^{\hbar} (z)\right)_{z=x+\hbar[\lambda  t+(1-\lambda) s]}d\lambda\\
&-\sqrt{\frac{\hbar N}{2}}\psi_{N}^{\hbar} (x+\hbar s)\int_0^1
\left(\sqrt{\frac{2N-2}{\hbar}}\psi_{N-2}^{\hbar} (z)-\frac{z}{\hbar}\psi_{N-1}^{\hbar} (z)\right)_{z=x+\hbar[\lambda  t+(1-\lambda) s]}d\lambda
\end{align*}
where we used relation~\eqref{eq:derivative_H} in the last equality.

We can now insert the uniform Plancherel-Rotach asymptotics~\eqref{eq:PR_bulk}-\eqref{eq:PR_out}, perform the integrals and use elementary trigonometric identities to conclude the proof of~\eqref{eq:proof1}. 

To prove the $C^1$-local uniform convergence,
 we start by taking the derivative(s) of the Christoffel-Darboux kernel $\partial^{\alpha}_t\partial^{\beta}_sK_{N,x,\hbar}\left(t,s\right)$. This entails computing the derivatives of Hermite functions. Now the trick is to write the derivative ${\psi_n^{\hbar}}'$ as a combination of Hermite functions (not differentiated) using again formula~\eqref{eq:derivative_H}. Hence, the local uniform asymptotics of ${\psi_n^{\hbar}}'$ can be read off from the Plancherel-Rotach asymptotics~\eqref{eq:PR_bulk}-\eqref{eq:PR_out} of $\psi_n^{\hbar}$. The proof of the $C^1$-convergence is therefore a simple modification of the proof of~\eqref{eq:proof1}.

To prove~\eqref{eq:proof2}, we  use again~\eqref{eq:derivative_H} to write the kernel as
\begin{align*}
K_{N}(x,y)&={\frac{\hbar }{2}}\frac{\psi_{N}^{\hbar} (x){\psi_{N}^{\hbar}}' (y)-\psi_{N}^{\hbar} (y){\psi_{N}^{\hbar}}'(x)}{x-y}-\frac{1}{2}\psi_{N}^{\hbar} (x)\psi_{N}^{\hbar} (y).
\end{align*}
If we set
\be
\Psi_N^{\hbar}(t):=\hbar^{-\frac{1}{6}}\left(V_{\sqrt{2\mu},\hbar^{\frac{2}{3}}}\psi_N^{\hbar}\right)(t)=\hbar^{\frac{1}{6}}\psi_{N}^{\hbar}(\sqrt{2\mu}+\hbar^{\frac{2}{3}}t),
\label{eq:Hermite_scaled_edge}
\ee
then,
\be
K_{N,\sqrt{2\mu},\hbar^{\frac{2}{3}}}\left(t,s\right)
={\frac{1 }{2}}\frac{\Psi_N^{\hbar}(t){\Psi_N^{\hbar}}'(s)-\Psi_N^{\hbar}(s){\Psi_N^{\hbar}}'(t)}{t-s}-\frac{\hbar^{\frac{1}{3}}}{2}\Psi_N^{\hbar}(t)\Psi_N^{\hbar}(s).
\label{eq:kernel_scaled_edge}
\ee
By the Plancherel-Rotach asymptotics~\eqref{eq:PR_edge}, for any compact set $J\Subset\C$,
\be
\lim_{N\to\infty}\sup_{t\in J}|\Psi_N^{\hbar}(t)-2^{\frac{1}{2}}c_{\mu}^{-\frac{1}{2}}\operatorname{Ai}(c_{\mu}t)|=0.
\ee

 Since the functions $\Psi_N^{\hbar}$ are entire, the above locally uniform convergence entails the uniform convergence of ${\Psi_N^{\hbar}}'$  to $\Ai'$ on compact subsets of $\C$ (a standard application of Cauchy's integral formula).

 By the very same argument, each finite-$N$ kernel $K_{N,\sqrt{2\mu},\hbar^{\frac{2}{3}}}$ is analytic, and hence their derivatives converge to the derivatives of the Airy kernel. The proof is complete.
\end{proof}

\begin{rem}Since $\sigmaPc\in \SS(\R_x\times\R_p)$, we have that the Fourier transform are rapidly decreasing functions too, $\FF_2\sigmaPc\in  \SS(\R_x\times\R_y)$. 
On the contrary,
$$
\FF_2\chi_D(x,y)=\mu\rho_{\mu}(x)K_{\sine}\left(-\mu\rho_{\mu}(x)y/2,\mu\rho_{\mu}(x)y/2\right)
=\frac{\sin\left[\sqrt{(2\mu-x^2)_+} y\right]}{\pi y}
$$ 
is not integrable in $\R_x\times\R_y$ and this tells that we cannot get in~\eqref{eq:convergence_K_N_K_sine} a convergence stronger than uniform on compact subsets.
\end{rem}
In order to prove Theorem~\ref{thm:convergenceAiry_intro} we will also need to  show that the Airy kernel on the antidiagonal is dominated by and integrable function.
\begin{lem} \label{lem:tail_K} There exist positive constants $C$, $c$ such that, for all $N\in\mathbb{N}$,  with $\hbar N$ fixed,
$$\left|K_{N,\sqrt{2\mu},\hbar^{\frac{2}{3}}}\left(-y,y\right)\right|\leq C e^{-c\left|y\right|^{\frac{3}{2}}},\quad \text{for all $y\in \R$}.$$
\end{lem}
\begin{proof} 
We start from the formula (recall the notation \eqref{eq:unitary_affine}):
\begin{multline}
K_{N,\sqrt{2\hbar N},\hbar^{\frac{1}{2}}\mu^{\frac{1}{6}}N^{-\frac{1}{6}}}(s,t)=\\
\sqrt{2} \mu^{\frac{1}{6}}N^{\frac{1}{3}}
\int_0^{\infty}
\left(V_{\sqrt{2\hbar N},\hbar^{\frac{1}{2}}\mu^{\frac{1}{6}}N^{-\frac{1}{6}}}\psi_N^{\hbar} \right)\left(u+s\right) 
\left(V_{\sqrt{2\hbar N},\hbar^{\frac{1}{2}}\mu^{\frac{1}{6}}N^{-\frac{1}{6}}}\psi_N^{\hbar}\right) \left(u+t\right) 
\,du\\
+ \frac{\mu^{\frac{1}{3}}}{2N^{\frac{1}{3}}}
\int_0^{\infty}(s+t+2u)
\left(V_{\sqrt{2\hbar N},\hbar^{\frac{1}{2}}\mu^{\frac{1}{6}}N^{-\frac{1}{6}}}\psi_N^{\hbar} \right)\left(u+s\right) 
\left(V_{\sqrt{2\hbar N},\hbar^{\frac{1}{2}}\mu^{\frac{1}{6}}N^{-\frac{1}{6}}}\psi_N^{\hbar}\right) \left(u+t\right) 
\,du\\
+\frac{1}{2}
\int_0^{\infty}\left(V_{\sqrt{2\hbar N},\hbar^{\frac{1}{2}}\mu^{\frac{1}{6}}N^{-\frac{1}{6}}}\psi_N^{\hbar} \right)'\left(u+s\right) 
\left(V_{\sqrt{2\hbar N},\hbar^{\frac{1}{2}}\mu^{\frac{1}{6}}N^{-\frac{1}{6}}}\psi_N^{\hbar}\right) \left(u+t\right) +\\
\bigl[\left(V_{\sqrt{2\hbar N},\hbar^{\frac{1}{2}}\mu^{\frac{1}{6}}N^{-\frac{1}{6}}}\psi_N^{\hbar} \right)\left(u+s\right) 
\left(V_{\sqrt{2\hbar N},\hbar^{\frac{1}{2}}\mu^{\frac{1}{6}}N^{-\frac{1}{6}}}\psi_N^{\hbar}\right)' \left(u+t\right) \bigr]
\,du.
\end{multline}
This is an identity true for all $\hbar$, $N$, and $\mu$. It can be proved from the representation~\eqref{eq:kernel_scaled_edge} using the differential equation for the Hermite functions.

If $\hbar N = \mu$, on the antidiagonal $-s=t=y$ we thus get
\begin{multline}
K_{N,\sqrt{2\mu},\hbar^{\frac{2}{3}}}(-y,y)=
\underbrace{
\sqrt{2} \mu^{\frac{2}{3}}
\int_0^{\infty} 
\Psi_N^{\hbar}(u-y)
\Psi_N^{\hbar} (u+y)
\,du}_{I_1}\\
+
\underbrace{\hbar^{\frac{2}{3}}
\int_0^{\infty} u\,
\Psi_N^{\hbar} (u-y)
\Psi_N^{\hbar}(u+y)
\,du}_{I_2}\\
+
\underbrace{\frac{\hbar^{\frac{1}{3}}}{2}
\int_0^{\infty}
\left[{\Psi_N^{\hbar}}'  (u-y)
\Psi_N^{\hbar} (u+y)
+
\Psi_N^{\hbar}  (u-y)
{\Psi_N^{\hbar}}'(u+y)\right]
\,du}_{I_3},
\end{multline}
where we used the notation~\eqref{eq:Hermite_scaled_edge}.
To estimate $I_1$, $I_2$, and $I_3$, we need some explicit bounds on the rescaled wavefunctions $\Psi_N^{\hbar}$.  Note that $K_{N,\sqrt{2\mu},\hbar^{\frac{2}{3}}}(-y,y)$ is even, and so it suffices to study the case $y>0$. 
A useful bound is
$$
\left|\psi_N^{\hbar}(y)\right|\leq \frac{C'}{N^{\frac{1}{12}}\hbar^{\frac{1}{4}}},
$$
for all $y$, see~\cite{Krasikov04}. Hence the rescaled wavefunctions are uniformly bounded by a constant
\be
\left|\Psi_N^{\hbar}(y)\right|\leq \frac{C'}{N^{\frac{1}{12}}\hbar^{\frac{1}{12}}}\leq C''.
\label{eq:Hermite_unif_bound}
\ee
To get an integrable estimate on $\Psi_N^{\hbar}(y)$ for $y>0$ we employ a theorem by Sonin and Polya~\cite[Theorem 7.31.1]{Szego39} giving quantitative growth information on the solutions of Sturm-Liouville equation. We observe that $\Psi_N^{\hbar}(y)$ satisfies the differential equation 
$$
{\Psi_N^{\hbar}}''=V\Psi_N^{\hbar},\quad  \text{where} 
\quad 
V(y)=2\sqrt{2\mu}y+\hbar^{\frac{2}{3}}y^2-\hbar^{\frac{1}{3}},
$$ 
for all $y\in\R$.  If $b$ denotes the positive zero of $V$, then we have
\begin{enumerate}
\item $\Psi_N^{\hbar}>0$ on $[b,+\infty)$;
\item $\lim_{y\to\infty}\left(\log\Psi_N^{\hbar}(y)\right)'=-\infty$.
\item $V>0$ and $V'>0$ on $[b,+\infty)$;
\end{enumerate} 
(For the first we use known bounds~\cite{Krasikov04} on the largest zero of the Hermite polynomial of degree $N$; the second is true because $\Psi_N^{\hbar}$ is a polynomial times a Gaussian.) The above mentioned theorem of Sonin and Polya (see the formulation in~\cite[Lemma 3.9.31]{Anderson10}) allows to conclude that
$$
\left(\log\Psi_N^{\hbar}(y)\right)'\leq -\sqrt{V}\quad \text{on $[b,+\infty)$}. 
$$
Hence,
\begin{align*}
\Psi_N^{\hbar}(y)&\leq \Psi_N^{\hbar}(b)\exp\left(-\int_{b}^y\sqrt{V(y')}dy'\right)\\
&\leq  \Psi_N^{\hbar}(b)\exp\left(-\int_{0}^y\sqrt{\left(2\sqrt{2\mu}y+\hbar^{\frac{2}{3}}y^2-\hbar^{\frac{1}{3}}\right)_+}dy'\right)\\
&\leq \Psi_N^{\hbar}(b)\exp\left(-\int_{0}^y\sqrt{\left(2\sqrt{2\mu}y -\hbar^{\frac{1}{3}}\right)_+}dy'\right)\\
&\leq c' \exp\left(-\frac{2}{3}c'\left(y-c'\hbar^{\frac{1}{3}}\right)^{\frac{3}{2}}\right),
\end{align*}
for all $y\geq b$. A short calculation shows that $0<b<\frac{1}{2}\frac{\hbar^{\frac{1}{3}}}{(2\mu)^{\frac{1}{2}}}+\frac{1}{8}\frac{\hbar^{\frac{4}{3}}}{(2\mu)^{\frac{3}{2}}}$.  Since $\Psi_N^{\hbar}(y)\to \Ai(y)$ pointwise, with different constants we have 
\be
\Psi_N^{\hbar}(y)\leq  c' \exp\left(-c'y^{\frac{3}{2}}\right), \quad\text{for $y\geq0$}.
\label{eq:Hermite_int_bound}
\ee
We can now estimate, for $y>0$,
\begin{align*}
I_1&\leq C
\int_0^{\infty} 
\left|\Psi_N^{\hbar}(u-y)\right|
\left|\Psi_N^{\hbar} (u+y)\right|
\,du\\
&\leq C\int_0^{+\infty}\left|\Psi_N^{\hbar}(u-y)\right|\exp\left(-c\left(u+y\right)^{\frac{3}{2}}\right)du\\
&\leq C\exp\left(-cy^{\frac{3}{2}}\right) \int_0^{+\infty}\left|\Psi_N^{\hbar}(u-y)\right|du\\
&\leq C\exp\left(-cy^{\frac{3}{2}}\right)\left( \int_0^{y}\left|\Psi_N^{\hbar}(u-y)\right|du+\int_y^{\infty}\left|\Psi_N^{\hbar}(u-y)\right|du\right)\\
&\leq C\exp\left(-cy^{\frac{3}{2}}\right)\left( \int_{-y}^{0}\Psi_N^{\hbar}(u)du+\int_0^{\infty}\Psi_N^{\hbar}(u)du\right)\\
&\leq C\exp\left(-cy^{\frac{3}{2}}\right)\left( C y+C\right)\\
&\leq C\exp\left(-cy^{\frac{3}{2}}\right).
\end{align*}
where  $C, c$ denote different constants in each line.
In the second to last step we used the uniform bound~\eqref{eq:Hermite_unif_bound} on $\R$ and the integrable bound~\eqref{eq:Hermite_int_bound} on $[0,\infty)$.

The analysis of $I_2$ and $I_3$ as functions of $y$ proceeds almost verbatim. Moreover, they are $o(1)$ as $N\to\infty$, so that their contribution is negligible.
\end{proof}

\begin{rem}
Since $K_{N,\sqrt{2\mu},\hbar^{\frac{2}{3}}}(-y,y)\to K_{\Ai}(-y,y)$ pointwise, it follows that $K_{\Ai}(-y,y)$ is also dominated by $Ce^{-c|y|^{\frac{3}{2}}}$. For an illustration of the kernels see Fig.~\ref{fig:KAiry}.
\end{rem}

\begin{figure}[t]
	\centering
	\includegraphics[width=1\textwidth]{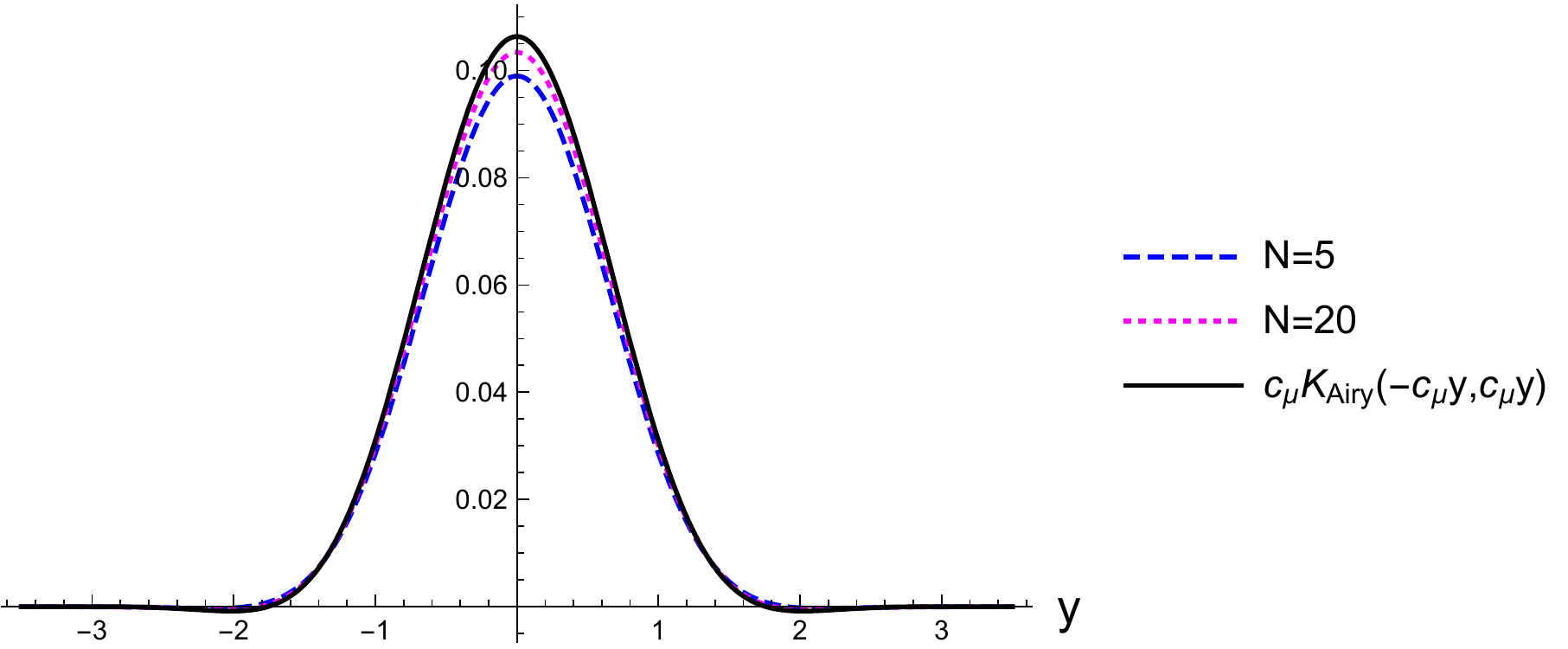}
	\caption{The kernel at the edge $K_{N,\sqrt{2\mu},\hbar^{\frac{2}{3}}}(-y,y)$ compared to its  limit $c_{\mu}K_{\Ai}(-c_{\mu}y,c_{\mu}y)$. Here $\mu=2$. Note the integrable tails (see Lemma~\ref{lem:tail_K}). }
	\label{fig:KAiry}
\end{figure}

\subsection{Asymptotics of the symbols}
Note that $\sigmaQc(x,p)\neq p\sigmaKc(x,p)$. This is \emph{not} surprising since the operators $\hat{p}$ and $P_{<N}$ do not commute. 
\begin{lem}
\label{lem:prod_symb}
For all $x,p \in \R$:
\begin{multline}
\sigmaQc(x,p)= \ p\sigmaKc(x,p)\\
+\frac{i}{2} \sqrt{\frac{\hbar N}{2}}\int_{\R_y}\hbar\left[\psi^{\hbar}_{N-1} \left(x-\frac{\hbar y}{2}\right) \psi^{\hbar}_{N}\left(x+\frac{\hbar y}{2}\right)-\psi^{\hbar}_{N} \left(x-\frac{ \hbar y}{2}\right) \psi^{\hbar}_{N-1} \left(x+\frac{\hbar y}{2}\right)\right]e^{ipy}dy.
\label{eq:sigmaQc-sigmaKc}
\end{multline}
\end{lem}
\begin{proof}
An application of the three-term recurrence of the Hermite functions~(\ref{eq:derivative_H}).
\end{proof}

\begin{prop}\label{prop:bound_symbols} The families $\{\sigmaKc\}_{N \geq 1}$ and $\{\sigmaQc\}_{N \geq 1}$ are bounded in $\AA'$.
Moreover, 
\be
\|\sigmaQc-p\sigmaKc\|_{\AA'}\leq \hbar\sqrt{\frac{\mu}{2}},
\label{eq:dist_lemma}
\ee
thus the distance between $\sigmaQc(x,p)$ and $p\sigmaKc(x,p)$ is  asymptotically small in $\AA'$, as $N\to\infty$, $\hbar \to 0$, with $\hbar N \to \mu>0$.
\end{prop}

\begin{proof}[Proof of Proposition~\ref{prop:bound_symbols}] Let $f\in\AA$. From Plancherel's theorem
\begin{align*}
\langle\sigmaPc,f\rangle
&
=\int_{\R_x\times\R_y} \overline{\hbar K_N\left(x-\frac{\hbar y}{2},x+\frac{\hbar y}{2}\right)}\FF_2f(x,y) dy dx,\\
\langle\sigmaHc,f\rangle
&
=\int_{\R_x\times\R_y}\overline{ \hbar Q_N\left(x-\frac{\hbar y}{2},x+\frac{\hbar y}{2}\right)}\FF_2f(x,y) dy dx.
\end{align*}
We can estimate 
\begin{align*}
\left|\langle\sigmaKc,f\rangle\right|&\leq\hbar\left(\int  \sup_x\left|\FF_2f(x,y) \right|dy\right)\left(\sup_y\int \left|K_N\left(x-\frac{\hbar y}{2},x+\frac{\hbar y}{2}\right) \right|dx\right)\\
&\leq\hbar  \|f\|_{\AA}\sup_y\sum_{j=0}^{N-1}\left(\int\left| \psi_k^{\hbar} \left( x-\hbar y/2\right)\psi_{k}^{\hbar}\left( x+\hbar y/2\right) \right|dx\right)\\
&\leq\hbar  \|f\|_{\AA}\sup_y\sum_{j=0}^{N-1}\left(\int \left| \psi_k^{\hbar} \left( x-\hbar y/2\right)\right|^2dx\,\int \left| \psi_k^{\hbar} \left( x+\hbar y/2\right)\right|^2dx\right)^{1/2}\\
&\leq\hbar  N \|f\|_{\AA}.
\end{align*}
Similarly, 
\begin{align*}
\left|\langle\sigma^\hbar_{H_N},f\rangle\right|
\leq \hbar\|f\|_{\AA}\sqrt{\frac{\hbar }{2}} \sup_y\sum_{j=0}^{N-2}2\sqrt{j+1}  
\leq \sqrt{2} (\hbar N)^{3/2}\|f\|_{\AA}.
\end{align*}
The convergent sequence $\hbar N$ is  bounded from above. The proof of the uniform boundedness of the symbols is complete.

With the help of Lemma~\ref{lem:prod_symb}, similar calculations are used in the proof of~\eqref{eq:dist_lemma},
\begin{align*}
&\|\sigmaQc-p\sigmaKc\|_{\AA'}\\
&\leq\sup_y\frac{\hbar}{2} \sqrt{\frac{\mu}{2}}\int_{\R}\left|\psi^{\hbar}_{N} \left(x-\frac{\hbar y}{2}\right) \psi^{\hbar}_{N-1}\left(x+\frac{\hbar y}{2}\right)-\psi^{\hbar}_{N-1} \left(x-\frac{\hbar y}{2}\right) \psi^{\hbar}_{N} \left(x+\frac{\hbar y}{2}\right)\right|dx\\
&\leq \hbar \sqrt{\frac{\mu}{2}}\|\psi_N^{\hbar}\|_2\|\psi_{N-1}^{\hbar}\|_2.
\end{align*}
\end{proof}

\subsection{Proofs of Theorems~\ref{thm:convergenceAprime_intro}, \ref{thm:convergenceAiry_intro} and~\ref{thm:edge_out_pointwise}}

\begin{proof}[Proof of Theorem~\ref{thm:convergenceAprime_intro}] 
Notice that, by Proposition~\ref{prop:Asymptotics_K_N}, for any compact sets $U,V\Subset\R$, 
there is a constant $C=C(U,V,\mu)>0$, such that
\begin{align}
\label{eq:estimate_conv}
&\sup_{x\in U}\sup_{ y\in V}\left|\left(\FF_2\sigmaKc-\FF_2\chi_{D}\right)\left(x,y\right)\right|\leq \frac{C}{N},\\
\label{eq:estimate_conv2}
&\sup_{x\in U}\sup_{ y\in V}\left|\left(\FF_2p \sigmaKc-\FF_2p\chi_{D}\right)\left(x,y\right)\right|\leq \frac{C}{N},
\end{align}
for all $N\geq1$, where 
$$
\FF_2\chi_{D}(x,y)=\mu\rho_{\mu}(x)K_{\sine}\left(-\mu\rho_{\mu}(x)y/2,\mu\rho_{\mu}(x)y/2\right).
$$
It is enough to show the two claims~\eqref{eq:asymp_sigma_P_<N}-\eqref{eq:asymp_sigma_H_N}  for all $f\in\BB$. By the density of $\BB$ in $\AA$ the thesis will follow. 
Let $f\in\BB$ (so that $\FF_2f$ has compact support $J\Subset\R_x\times\R_y$, see Section~\ref{sec:notation}). Then,
\begin{align*}
&\left|\int\overline{ \left[\sigmaPc\left(x,p\right)-\chi_{D}(x,p)\right]}f(x,p)dx dp\right|\\
&=\left|\int
\overline{\left[\FF_2\sigmaKc(x,y)-\FF_2\chi_{D}(x,y)\right]}\FF_2f(x,y)dxdy\right|\\
&\leq\|\FF_2f\|_{\infty}\int_J\left|\FF_2\sigmaKc(x,y)-\FF_2\chi_{D}(x,y)\right|dxdy\leq C\|\FF_2f\|_{\infty}N^{-1}.
\end{align*}
for some constant $C$ (dependent on $J$). Hence, for all $f\in\BB$, 
$$
\lim_{N\to\infty}\langle\sigmaKc-\chi_D,f\rangle=0.
$$
Similarly, for $f\in\BB$,
$$
\left|\langle\sigmaQc-p\chi_D,f\rangle\right|\leq \left|\langle \sigmaQc- p\sigmaKc,f\rangle\right|+\left|\langle p\sigmaKc-p\chi_D,f\rangle\right|.
$$ 
The first term is of order $O(N^{-1})$ by Proposition~\ref{prop:bound_symbols}; the second term is also $O(N^{-1})$ by Eq.~\eqref{eq:estimate_conv2}.
Hence, for all $f\in\BB$, 
$$
\lim_{N\to\infty}\langle\sigmaQc-p\chi_D,f\rangle=0.
$$
\end{proof}
Before proving Theorem~\ref{thm:convergenceAiry_intro}, we need some notation. Consider the change of coordinates
    \begin{align}
    T\colon \R_x\times \R_p &\to \C^2 \nonumber \\
(x,p)&\mapsto (\theta,\zeta)
    \end{align}
where $\zeta$ is solution of
\begin{align}
x^2+p^2=2\mu+\zeta^2,
\end{align}
and $\theta\in[0,2\pi)$ is given by
\be
\theta=
\begin{cases}
\arctan\frac{p}{x} &\text{if $x\neq0$}\\
\frac{\pi}{2} &\text{if $x=0$, $p>0$}\\
\frac{3\pi}{2} &\text{if $x=0$, $p<0$}
\end{cases}.
\ee
If  $(x,p)\notin D$ then $\zeta\in(0,\infty)$; 
if t$(x,p)\in D$, then $\zeta\in[0,i\sqrt{2\mu}]$. $T$ is a bijection from $\R_x\times \R_p$  to $T(\R^2\setminus(0,0))=[0,2\pi)\times (\R_{+}\cup [0,i\sqrt{2\mu}])$ with Jacobian determinant 
\be 
|J|:=\left|\det\frac{\partial T(x,p)}{\partial (\zeta,\theta)}\right|= |\zeta|.
\ee
We can now prove Theorem~\ref{thm:convergenceAiry_intro}.
\begin{proof}[Proof of Theorem~\ref{thm:convergenceAiry_intro}] 
Let $g\in C^{\infty}_c(\mathbb{R})$.  We have
\begin{align*}
&\int_{\R_x\times\R_p}\sigmaKc(x,p)\frac{1}{\hbar^{\frac{2}{3}}}g\left(\frac{x^2+p^2-2\mu}{\hbar^{\frac{2}{3}}}\right)dxdp\\
&=\int_{D^c}\sigmaKc(x,p)\frac{1}{\hbar^{\frac{2}{3}}}g\left(\frac{x^2+p^2-2\mu}{\hbar^{\frac{2}{3}}}\right)dxdp+\int_{D}\sigmaKc(x,p)\frac{1}{\hbar^{\frac{2}{3}}}g\left(\frac{x^2+p^2-2\mu}{\hbar^{\frac{2}{3}}}\right)dxdp\\
&=2\pi \int_{0}^{+\infty}\sigmaKc(\sqrt{2\mu},z\hbar^{\frac{1}{3}})zg(z^2)dz+2\pi \int_{0}^{+\infty}\sigmaKc(\sqrt{2\mu},iz\hbar^{\frac{1}{3}})zg(-z^2)\chi_{(0,\frac{2\mu}{\hbar^{2/3}})}(z)dz\\
&=2\pi  \int_{\R_y}\int_0^{+\infty}K_{N,\sqrt{2\mu},\hbar^{\frac{2}{3}}}\left(-y/2,y/2\right)e^{-izy}zg(z^2)dzdy\\
&+2\pi  \int_{\R_y}\int_0^{+\infty}K_{N,\sqrt{2\mu},\hbar^{\frac{2}{3}}}\left(-y/2,y/2\right)e^{zy}zg(-z^2)\chi_{(0,\frac{2\mu}{\hbar^{2/3}})}(z)dzdy,
\end{align*}
where in the third line we used the rotational symmetry of the symbol $\sigmaKc$ and we considered the symbol $\sigmaKc$ as a function on $\C_x\times\C_p$ (see Remark~\ref{rem:complex}). Similarly,
\begin{align*}
&\int_{\R_x\times\R_p}\chi_{D}^{(N)}\left(x,p\right)\frac{1}{\hbar^{\frac{2}{3}}}g\left(\frac{x^2+p^2-2\mu}{\hbar^{\frac{2}{3}}}\right)dxdp\\
&=2\pi\int_{0}^{+\infty}\Ai_1\left(\frac{1}{2^{\frac{1}{3}}\mu^{\frac{1}{3}}}\left(\frac{x^2+p^2-2\mu}{\hbar^{\frac{2}{3}}}\right)\right)\frac{1}{\hbar^{\frac{2}{3}}}g\left(\frac{x^2+p^2-2\mu}{\hbar^{\frac{2}{3}}}\right)dxdp\\
&=2\pi \int_{\R_y}\int_0^{+\infty}c_{\mu}K_{\Ai}(-c_{\mu}y/2,c_{\mu}y/2)e^{-izy}zg(z^2)dzdy\\
&+2\pi \int_{\R_y}\int_0^{+\infty}c_{\mu}K_{\Ai}(-c_{\mu}y/2,c_{\mu}y/2)e^{zy}zg(-z^2)\chi_{(0,\frac{2\mu}{\hbar^{2/3}})}(z)dzdy.
\end{align*}
Hence, 
\begin{align*}
&\left|\int_{\R_x\times\R_p}\left[\sigmaKc(x,p)-\chi_{D}^{(N)}\left(x,p\right)\right]\frac{1}{\hbar^{\frac{2}{3}}}g\left(\frac{x^2+p^2-2\mu}{\hbar^{\frac{2}{3}}}\right)dxdp\right|\leq
I_N+J_N,
\end{align*}
where
\begin{align*}
I_N=2\pi
\int_0^{+\infty}\left(\int_{\R_y}f_N(y)dy\right)\left|zg(z^2)\right|dz,\quad
J_N=2\pi
\int_0^{+\infty}\left(\int_{\R_y}f_N(y)e^{zy}dy\right)\left|zg(z^2)\right|dz,
\end{align*}
and $f_N(y):=\left|K_{N,\sqrt{2\mu},\hbar^{\frac{2}{3}}}\left(-y/2,y/2\right)-c_{\mu}K_{\Ai}(-c_{\mu}y/2,c_{\mu}y/2)\right|$. By Proposition~\ref{prop:Asymptotics_K_N2} $f_N(y)$ tends to zero uniformly on compact sets (and hence pointwise). 
The tail estimate of Lemma~\ref{lem:tail_K} implies that the sequences $f_N(y)$ and $f_N(y)e^{zy}$ are dominated by an integrable function, and so by the dominated convergence theorem  both $\int f_N(y)dy$ and $\int f_N(y)e^{zy}dy$ tend to zero (the latter for any $z$).  Since $\operatorname{supp}g\Subset\mathbb{R}_z$, we conclude that both $I_N$ and $J_N$ go to zero as $N\to\infty$. This proves~\eqref{eq:edge_asymp_K} from which~\eqref{eq:edge_asymp_Q} follows by recalling Lemma~\ref{lem:prod_symb}.
\end{proof}
\begin{proof}[Proof of Theorem~\ref{thm:edge_out_pointwise}] It is again enough to prove~\eqref{eq:edge_pointwise1}. The second claim~\eqref{eq:edge_pointwise2} will follow by Lemma~\ref{lem:prod_symb}. Fix $\epsilon>0$. By rotational symmetry we can assume $x=\sqrt{2\mu}$ and $p=z\in\R$.
\begin{align*}
 \left|\sigmaKc(\sqrt{2\mu},z)-\chi_{D}^{(N)}(\sqrt{2\mu},z)\right|&\leq\int_{\R_y}\left| K_{N,\sqrt{2\mu},\hbar^{\frac{2}{3}}}\left(y/2,-y/2\right)-c_{\mu}K_{\Ai}(-c_{\mu}y/2,c_{\mu}y/2)\right|dy\\
 &\leq \int_{-L}^L\left| K_{N,\sqrt{2\mu},\hbar^{\frac{2}{3}}}\left(-y/2,y/2\right)-c_{\mu}K_{\Ai}(-c_{\mu}y/2,c_{\mu}y/2)\right|dy \\
 &+2\int_{L}^{+\infty}\left| K_{N,\sqrt{2\mu},\hbar^{\frac{2}{3}}}\left(-y/2,y/2\right)\right|dy\\
 &+2\int_{L}^{+\infty}\left|c_{\mu}K_{\Ai}(-c_{\mu}y/2,c_{\mu}y/2)\right|dy:=I_1+I_2+I_3
 \end{align*}
 for every $L>0$. Choose $L_0$ such that the $I_2$ and $I_3$ are each bounded by $\epsilon/3$.  The first integral  $I_1$ is bounded by  $C\hbar^{\frac{1}{3}}$, with $C=C(L_0)$. Take $\hbar= \left(\epsilon/(3C)\right)^3$ and conclude the proof.
\end{proof}

\section*{Acknowledgements}

FDC thanks Nick Simm for helpful correspondence. The authors would also like to thank Gerardo Hernandez-Duenas and Alejandro Uribe for pointing out the paper~\cite{Hernandez-Duenas15}.
We acknowledge the support by the Italian National Group of Mathematical Physics (GNFM-INdAM), by PNRR MUR projects CN00000013-`Italian National Centre on HPC, Big Data and Quantum Computing' and PE0000023-NQSTI,  by Regione Puglia through the project `Research for Innovation' - UNIBA024, and by Istituto Nazionale di Fisica Nucleare (INFN) through the project `QUANTUM'.

\section*{Declarations}
\subsection*{Funding and/or Conflicts of interests/Competing interests} The authors have no competing interests to declare that are relevant to the content of this article.

\subsection*{Data availability} 
Data sharing not applicable to this article as no datasets were generated or analysed during the current study.

\appendix

\section{Quantum harmonic oscillator and  Hermite polynomials}
\label{app:A}

The classical harmonic oscillator Hamiltonian function is  
\be
\mathfrak{h}_{\operatorname{h.o.}}(x,p)=\frac{1}{2}\left(p^2+x^2\right).
\ee
The \emph{harmonic oscillator Schr\"odinger operator} is
\be
H_{\operatorname{h.o.}}=\frac{1}{2}\left(-\hbar^2\partial^2_x+\hat{x}^2\right).
\ee

The \emph{Hermite functions} $(\alpha^2=1/\hbar)$ are
\be
\psi_k^{\hbar} (x)=\sqrt{\frac{\alpha }{\sqrt{\pi } 2^k k!}} \exp \left(-\frac{1}{2}\alpha ^2 x^2\right) h_k(\alpha  x),\qquad k=0,1,2,\ldots
\ee
where
\be
h_k(y)=(-1)^ke^{y^2}\frac{d^k}{dy^k}e^{-y^2}
\ee
is the $k$-th \emph{Hermite polynomials}. The Hermite functions are eigenfunctions of the harmonic oscillator
\be
H_{\operatorname{h.o.}}\psi_k^{\hbar} (x)=\lambda_k\psi_k^{\hbar} (x),\quad k=0,1,2,\ldots,
\ee
with eigenvalues $\lambda_k=\hbar \left(k+\frac{1}{2}\right)$, and form an orthonormal basis in $L^2(\R)$
\be
\int_{\R} \psi_k^{\hbar} (x)\psi_{\ell}^{\hbar} (x)dx=\delta_{k,\ell}.
\ee
Useful formulae are the following three-term relations written in terms of position operator $\hat{x}$ and momentum operator $\hat{p}=-i\hbar\frac{d}{d x}$:
\begin{align}
(\hat{x}\psi_k^{\hbar}) (x)&=\sqrt{\frac{\hbar}{2}}\left[\sqrt{k+1}\psi_{k+1}^{\hbar} (x)+\sqrt{k}\psi_{k-1}^{\hbar} (x)\right]\label{eq:Hermite_x}
\\
(\hat{p}\psi_k^{\hbar} )(x)&=i\sqrt{\frac{\hbar}{2}}\left[\sqrt{k+1}\psi_{k+1}^{\hbar} (x)-\sqrt{k}\psi_{k-1}^{\hbar} (x)\right].
\label{eq:Hermite_p}
\end{align}
When combined they give the useful relation
\be
\label{eq:derivative_H}
\frac{d}{dz}{\psi_{k}^{\hbar} }(z)=\sqrt{\frac{2k}{\hbar}}\psi_{k-1}^{\hbar} (z)-\frac{z}{\hbar}\psi_{k}^{\hbar} (z).
\ee

We have the following Plancherel-Rotach asymptotics formulae (see~\cite[Theorem 8.22.9]{Szego39}). 
Let $\epsilon<\epsilon'$ be fixed positive numbers, and $n\in\Z$ fixed. Let $\hbar=\hbar_N$ so that $\hbar_N N=\mu$ , where $\mu>0$ is a fixed number. The following asymptotics hold true:
\begin{enumerate}[(i)]
\item If $x=\sqrt{(2+1/N)\mu}\cos\phi$, $\epsilon\leq\phi\leq\pi-\epsilon$, then
\be
\label{eq:PR_bulk}
\psi_{N+n}^{\hbar} \left(x\right)= \left(\frac{2}{\mu}\right)^{\frac{1}{4}}\left(\frac{1}{\pi \sin\phi  }\right)^{\frac{1}{2}} \left\{ \sin \left[\left(\frac{N}{2}+\frac{1}{4}\right) (\sin (2 \phi )-2 \phi )+\frac{3 \pi }{4}-n\phi\right]+O(N^{-1})\right\};
\ee
\item If $x=\sqrt{(2+1/N)\mu}\cosh\phi$, $\epsilon\leq\phi\leq\epsilon'$, then
\be
\label{eq:PR_out}
\psi_{N+n}^{\hbar} \left(x\right)=\left(\frac{1}{8\mu}\right)^{\frac{1}{4}} \left(\frac{1}{\pi  \sinh\phi }\right)^{\frac{1}{2}} \exp \left(-\left(\frac{N}{2}+\frac{1}{4}\right) (\sinh (2 \phi )-2 \phi )+n\phi\right)\left(1+O(N^{-1})\right);
\ee
\item
If $x=\sqrt{(2+1/N)\mu}+\sqrt{\mu/2}{N^{-\frac{2}{3}}}t$, with $t$ complex and bounded, then,
\be
\label{eq:PR_edge}
\psi_N^{\hbar} \left(x\right)=\left(\sqrt{2/\mu}N^{1/3}\right)^\frac{1}{2}\operatorname{Ai}(t)+O(N^{-1/2}).
\ee
\end{enumerate}
In all these formulae, the $O$-terms hold uniformly.
Note that the choice $\hbar N=\mu$ is the right scaling of a vanishing Planck constant that gives rise to a nontrivial asymptotics.

\section{Sine kernel, Airy kernel and their Fourier transforms}
\label{app:B}
We denote the normalised \emph{semicircular density} of radius $\sqrt{2\mu}>0$,
\be
\label{eq:semic}
\rho_{\mu}(x)=\frac{1}{\pi\mu}\sqrt{(2\mu-x^2)_+},
\ee
and 
\be
\label{eq:const}
c_{\mu}=2^{\frac{1}{2}}\mu^{\frac{1}{6}}.
\ee
Here is why these factors pop out in all formulae. If the energy of a (classical) harmonic oscillator is $\mathfrak{h}_{\mathrm{h.o.}}(x,p)=\mu$, then the momentum as a function of position is $p(x)=\sqrt{(2\mu-x^2)_+}$. At the points of inversion of motion $x=\pm\sqrt{2\mu}$, the momentum is zero, and its one-side derivative is $|p'(\pm\sqrt{2\mu})|=c_{\mu}^{3}$.

The sine and the Airy kernels are
\begin{align}
\label{eq:sinek} K_{\sine}(u,v)&=\frac{\sin \pi \left(u-v\right)}{\pi (u-v)}
 \quad&&\text{(\emph{sine kernel})},\\
 \label{eq:Airyk}
  K_{\Ai}(u,v)&=\frac{\operatorname{Ai}(u)\operatorname{Ai}'(v)-\operatorname{Ai}'(u)\operatorname{Ai}(v)}{u-v}\quad&&\text{(\emph{Airy kernel})},
\end{align}
where the \emph{Airy function} is defined by the formula 
\be
\Ai(x):=\frac{1}{2\pi i}\int_C e^{\zeta^3/3-x\zeta}d\zeta,
\ee
with $C$ a contour in the complex $\zeta$-plane consisting of the ray joining $e^{-i\pi/3}{\infty}$ to the origin plus the ray joining the origin to $e^{i\pi/3}{\infty}$.

The kernels~\eqref{eq:sinek}-\eqref{eq:Airyk} are defined for $u=v$ in the unique way making them continuous (and in fact $C^{\infty}$).
The sine kernel can be viewed as the `square' of another symmetric kernel,
\be
 K_{\sine}(u,v)=\int_{-1}^{1}e^{\pi i  u\lambda}\overline{e^{\pi i  v\lambda}}d\lambda.
\ee 
A similar identity holds for the Airy kernel,
\be
K_{\Ai}(u,v)=\int_0^{+\infty}\operatorname{Ai}(u+\lambda)\operatorname{Ai}(v+\lambda)d\lambda.
\label{eq:int_rep_Ai}
\ee
(Use~\eqref{eq:commutator_DL} and the Airy differential equation, see ~\cite{TW94}.)
Using a trick one gets the following useful representation:
\be
K_{\Ai}(u,v)=\int_{\R}e^{iq(t-s)}\left(\int_{0}^{+\infty}\operatorname{Ai}\left(\lambda+2^{2/3}q^2+(u+v)/2^{1/3}\right)d\lambda\right)\frac{dq}{2\pi}.
\label{eq:AiryK_trick}
\ee
Note that $K_{\sine}(u,-u)$ is locally integrable (but not in $L_1(\R)$), while $\int_{\R} \left|K_{\Ai}(u,-u)\right|du<\infty$.

$K_{\sine}$ is the Fourier transform of the characteristic function of the disk:
\be
\label{eq:Fourier_chiD}
\int_{\R}\mu\rho_{\mu}(x)K_{\sine}(-\mu\rho_{\mu}(x)y/2,\mu\rho_{\mu}(x)y/2)e^{ipy}dy=\chi_{D}(x,p).
\ee
From~\eqref{eq:AiryK_trick} we also get an explicit formula for the Fourier transform of $K_{\Ai}$:
\be
\int_{\R}c_{\mu}K_{\Ai}(-c_{\mu}y/2,c_{\mu}y/2)e^{izy}dy=\Ai_1\left(\frac{z^2}{\left(2\mu\right)^{\frac{1}{3}}}\right),
\ee
where 
\be
\label{eq:Ai1}
\Ai_1(\xi):=\int_{\xi}^{+\infty}\operatorname{Ai}\left(u\right)du
\ee 
is the \emph{integrated Airy function}. 
We have $\Ai_1(-\infty)=\int_{\R}\operatorname{Ai}\left(u\right)du=1$. The function $\Ai_1(\xi)$ has the following large $|\xi|$ asymptotics~\cite[Eq. (9.10.4)-(9.10.6)]{NIST}:
\be
\label{eq:asymp_Ai1}
\Ai_1(\xi)\sim
\begin{cases}
\dfrac{1}{2\pi^{1/2}|\xi|^{3/4}}e^{-\frac{2}{3}|\xi|^{2/3}}, &\text{for $\xi\to+\infty$},\\\\
1-\dfrac{1}{\pi^{1/2}|\xi|^{3/4}}\cos\left({\frac{2}{3}|\xi|^{2/3}+\frac{\pi}{4}}\right), &\text{for $\xi\to-\infty$}.
\end{cases}
\ee

\end{document}